\documentclass{IEEEtran}
\usepackage{amsmath,amsfonts}
\usepackage{amssymb}
\usepackage{cite}
\usepackage{algorithmicx}
\usepackage{algorithm}
\usepackage{algpseudocode}
\usepackage{graphicx}
\usepackage{subfigure}
\usepackage{makecell}
\usepackage{multirow}
\usepackage{multicol}
\usepackage{cases}
\usepackage{stfloats}
\usepackage{enumerate}
\usepackage{url}
\usepackage{booktabs}
\usepackage{diagbox}
\usepackage{bbding}
\usepackage{color}
\usepackage{xcolor}
\usepackage{framed}
\usepackage{amsthm}
\usepackage[normalem]{ulem}

\newtheorem{theorem}{Theorem}
\newtheorem{lemma}{Lemma}
\allowdisplaybreaks[4]

\makeatletter
\renewcommand{\maketag@@@}[1]{\hbox{\m@th\normalsize\normalfont#1}}
\makeatother



\begin{document}

\title{Non-Orthogonal Multiple-Access for Coherent-State Optical Quantum Communications Under Lossy Photon Channels}

\author{Zhichao Dong, Xiaolin Zhou, \IEEEmembership{Senior Member, IEEE,} Yongkang Chen,  \\
Wei Ni, \IEEEmembership{Fellow, IEEE,} Ekram Hossain, \IEEEmembership{Fellow, IEEE,} and Xin Wang, \IEEEmembership{Fellow, IEEE}
\thanks{Z. Dong, X. Zhou, and X. Wang are with the Key Laboratory for Information Science of Electromagnetic Waves, School of Information Science and Technology, Fudan University, Shanghai 200433, China (e-mail: zhichaodong@ynu.edu.cn; zhouxiaolin@fudan.edu.cn; xwang11@fudan.edu.cn)}
\thanks{Y. Chen is with the College of Information Engineering, Shanghai Maritime University, Shanghai 201306, China (e-mail: ykchen@shmtu.edu.cn).}
\thanks{W. Ni is with Data61, Commonwealth Scientific and Industrial Research Organisation, Sydney, NSW 2122, Australia, and the School of Computer Science and Engineering, University of New South Wales, Sydney, NSW 2052, Australia (e-mail: wei.ni@ieee.org).}
\thanks{Ekram Hossain is with the Department of Electrical and Computer Engineering, University of Manitoba, Winnipeg, MB R3T 2N2, Canada (e-mail:
ekram.hossain@umanitoba.ca).}
}

\IEEEpubid{}

\maketitle

\begin{abstract}
Coherent states have been increasingly considered in optical quantum communications (OQCs). With the inherent non-orthogonality of coherent states, non-orthogonal multiple-access (NOMA) naturally lends itself to the implementation of multi-user OQC. However, this remains unexplored in the literature.
This paper proposes a novel successive interference cancellation (SIC)-based photon-number-resolving detection (PNRD)-Kennedy receiver for uplink NOMA-OQC systems, along with a new approach for power allocation of the coherent states emitted by users.
The key idea is to rigorously derive the asymptotic sum-rate of the considered systems, taking into account the impact of atmospheric turbulence, background noise, and lossy photon channel. 
With the asymptotic sum-rate, we optimize the average number of photons (or powers) of the coherent states emitted by the users. Variable substitution and successive convex approximation (SCA) are employed to convexify and maximize the asymptotic sum-rate iteratively. 
A new coherent-state power allocation algorithm is developed for a small-to-medium number of users. 
We further develop its low-complexity variant using adaptive importance sampling, which is suitable for scenarios with a medium-to-large number of users. Simulations demonstrate that our algorithms significantly enhance the sum-rate of uplink NOMA-OQC systems using coherent states by over 20\%, compared to their alternatives.
\end{abstract}


\section{Introduction}
\IEEEPARstart{Q}{uantum} communication has been increasingly considered for future communications due to its high transmission data rate and reliability \cite{Survey-zhou,Quantum-mechanical,Quantum_Limits}. Coherent states (Glauber states) are a special type of quantum state in quantum optics, characterized by minimal fluctuations in the number of photons.
The output of practical laser sources can be accurately modeled as coherent states, enabling compatibility with classical optical communication systems.
Compared to non-coherent states, coherent states are advantageous for practical implementations, offering stable properties with well-defined amplitudes and phases \cite{Glauber}. Their compatibility with classical systems and the maturity of laser technology have established coherent states as ideal information carriers in optical quantum communication (OQC) systems \cite{FSO-nature, Quantum_Communication,Multiaccess,Diversity}. The reason is that single-antenna configurations remain predominant in OQC systems due to their inherent advantages in preserving quantum state integrity~\cite{background-noise} and system simplicity \cite{Threshold}. The direct alignment between single-antenna outputs and quantum receivers also eliminates multi-path interference and enables efficient quantum state detection—an essential requirement for OQC systems~\cite{FSO-nature}. By contrast, deployment of multi-antenna systems still faces significant challenges, including reduced detection efficiency, lower state fidelity, and increased receiver design complexity.

The study in \cite{FSO-turbulence} demonstrated that increasing atmospheric turbulence significantly raises signal attenuation and quantum bit error rate (QBER), with the secure key rate dropping to zero when the Rytov variance exceeds~1.3. In a real-world satellite-to-ground experiment, the authors of \cite{FSO-nature} observed turbulence-induced arrival time jitter and mode distortion, requiring adaptive optics and strict filtering to maintain secure operation. It was confirmed that higher background noise levels during daytime would exceed the QBER security threshold, leading to key distribution failure. The authors of~\cite{background-noise} showed that solar scattering during daytime severely degrades key rates, while nighttime operation improved performance by two orders of magnitude. 
Moreover, the study in \cite{Nature2} quantitatively analyzed photon loss, demonstrating that it directly increases QBER and constrains system performance, particularly under extremely weak signal conditions.

Coherent states possess inherent non-orthogonal properties, making non-orthogonal multiple access (NOMA) a natural approach for implementing multi-user OQC to enhance system capacity. However, NOMA-OQC systems employing coherent states remain unexplored in the literature and present significant research challenges~\cite{Multiaccess, Multiple-Access}. In practical OQCs,
coherent states are affected by atmospheric turbulence~\cite{FSO-turbulence,FSO-nature}, background noise~\cite{background-noise}, and photon loss during detection~\cite{Nature2}. These practical constraints create unique challenges in the allocation of the powers of the coherent states emitted by the users in NOMA-OQC systems. Moreover, the existing techniques developed for classical communication systems are not applicable due to critical distinctions between OQC systems employing coherent states and classical additive white Gaussian noise (AWGN) systems; see Table~\ref{tab:Comparison_AWGN}.

\subsection{Related Work}
In the context of multiple-access OQC systems, Concha \textit{et al.}\cite{Multiaccess} introduced a method for quantum multiple-access using coherent states in both single-mode and multi-mode scenarios. Chou \cite{Multiple-Access} investigated the capacity region of quantum multiple-access channels, while the authors of \cite{CDMA} constructed a quantum code-division multiple access (CDMA) communication system based on quantum multiple-access technology. Moreover, advanced channel coding schemes, such as Low-Density Parity-Check (LDPC) codes \cite{LDPC},  polar codes \cite{Polar-code}, and turbo codes \cite{Turbo-codes}, have demonstrated their effectiveness in mitigating error propagation in the classical NOMA systems. This provides a theoretical foundation for SIC-based multi-user detection in NOMA-OQC systems.
Furthermore, the authors of \cite{Weng} proposed a non-orthogonal quantum multi-user iterative detection scheme, where users emit coherent states and effective interference suppression is demonstrated.

On the other hand, in OQC systems, the loss of photons in the coherent states can be modeled by a lossy photon channel \cite{lossy,lossy_photon2}. Łukanowski and Jarzyna \cite{lossy_photon2}
calculated the capacity of a lossy photon channel based on the average number of photons detected in coherent states using photon-number-resolving detection (PNRD) \cite{lossy}. Building on this, a weak-field homodyne receiver with finite-resolution PNRDs was later developed as a superior practical alternative to conventional homodyne and direct detection schemes for coherent states transmitted over lossy channels \cite{weak-field}. PNRD-based detection significantly improves quantum state reception in lossy channels due to its quantum-limited photon-number sensitivity (following Poisson statistics) and enhanced communication performance with significant input signals~\cite{PNRD_2,Threshold}. As established in \cite{Finite-PNRD1,Finite-PNRD2}, PNRD-based quantum receivers achieve a performance advantage over both the standard quantum limit (SQL) and classical Kennedy receivers, even under practical constraints such as finite resolution and non-ideal detection efficiency. In this work, we refer to a quantum receiver architecture combining displacement operation with PNRD as the PNRD-Kennedy receiver.

Using coherent states as information carriers in OQC systems has received growing attention~\cite{detection,Threshold,Diversity,Performance,Binary}. Taking background noise into consideration, Semenov \textit{et al.}~\cite{detection} established a signal-dependent noise model to characterize the influence of background noise. Considering atmospheric turbulence, Yuan \textit{et al.} \cite{Diversity} proposed an optical combining technique for receiving coherent states in OQC systems. In another study, Yuan \textit{et al.} \cite{Threshold} analytically studied the performance of the Kennedy receiver and proposed the optimally displaced threshold detection under noises. Vázquez-Castro \textit{et al.}~\cite{Binary} verified the performance advantages of BPSK modulation of coherent states in OQC systems compared to classical communication systems under turbulent channels.

These existing studies, i.e., \cite{Threshold,Diversity,Performance,Binary,lossy},
have focused primarily on studying the performance of OQCs, including optimal threshold detection, quantum bit error rate, and system capacity. There is a lack of studies on the power allocation of coherent states in multi-user OQC systems over a lossy photon channel. Let alone NOMA-OQC.
\begin{table}
  \centering
  \renewcommand{\arraystretch}{1}
  \caption{Comparison between OQC  systems employing coherent states and classical AWGN systems }
  \label{tab:Comparison_AWGN}
  \begin{tabular}{|p{0.15\linewidth}|p{0.3\linewidth}|p{0.3\linewidth}|c|}
    \hline
     & \textbf{Coherent-state OQC systems} & \textbf{Classical AWGN systems}\\
    \hline
    \textbf{Information carrier}  & Coherent states & Radio frequency electromagnetic waves\\
    \hline
     \textbf{Noise} & Signal-dependent noise  \cite{detection}& Additive white Gaussian noise\\
     \hline
     \textbf{Received signal} & Poisson-distributed \cite{lossy} & Gaussian-distributed\\
     \hline
     \textbf{Sum-rate calculation} & Based on the definition of mutual information & Based on Shannon formula\\
     \hline
  \end{tabular}
\end{table}
\subsection{Contribution}
In this paper, we develop sum-rate analysis and propose a novel communication resource allocation approach for uplink NOMA-OQC systems with coherent-state users based on a successive interference cancellation (SIC)-enhanced PNRD-Kennedy receiver. The contributions are as follows.  

\begin{itemize}
\item We propose a novel SIC-based PNRD-Kennedy receiver for uplink NOMA-OQC systems that explicitly accounts for the impact of atmospheric turbulence, background noise, and photon loss.  

\item We analytically derive the sum-rate of the uplink NOMA-OQC systems.
We establish the upper and lower bounds of the sum-rate, prove their asymptotic tightness, and obtain an asymptotic approximation of the sum-rate.

\item We design a new coherent-state power allocation algorithm by iteratively convexifying the asymptotic sum-rate using successive convex approximation (SCA).
The algorithm converges to a quality solution satisfying the Karush-Kuhn-Tucker (KKT) conditions.

\item We also develop a low-complexity coherent-state power allocation algorithm using adaptive importance sampling.
Polynomial complexity is achieved, while the KKT conditions are preserved upon convergence. This algorithm is suited under a medium-to-large number of users. 
\end{itemize}

As corroborated by extensive simulations, the asymptotic sum-rate is tight, and the proposed coherent-state power allocation algorithms can substantially outperform their possible alternatives by 20\% in the achievable sum-rate. The proposed low-complexity algorithm can alleviate the computational demands and, hence, substantially enhance the scalability in the considered NOMA-OQC systems.

The rest of this paper is organized as follows.  Section~\ref{sec:SysModel} establishes the system model, Section~\ref{sec:Analysis} develops the sum-rate analysis, Section~\ref{sec:Algorithm} proposes the power allocation scheme, and Section~\ref{sec:Low-Complexity} proposes its low-complexity implementation. Simulation results are presented in Section~\ref{sec:Simulation}, followed by conclusions in Section~\ref{sec:Conclusion}.

\begin{table}[t]
  \centering
  \renewcommand{\arraystretch}{1}
  \caption{Notation and Definitions}
  \label{tab:notations}
  \begin{tabular}{cp{0.7\linewidth}}
    \hline
    \textbf{Notation} & \textbf{Description} \\
    \hline
    $X_k$ & BPSK signal of the $k$-th user\\
    $n _b$  & Background noise \\
$\left\{ \hat{\Pi}_y \right\} _{y=0}^{\infty}$  & POVM measurement operators \\
$\eta$  & Transmittance coefficient of the lossy photon channel \\
$\tau_k$ & The power allocation factor of user $k$\\
$\phi_k$  & The amplitude of received coherent state of user $k$ \\
$h_k$  &  Fading coefficient \\
$h_{k,t}$  & Transmissivity of the atmospheric turbulence \\
$h_{k,l}$  &  Path loss \\
$\boldsymbol{X}_i$ & Enumerate all $2^K$ possibilities of $\boldsymbol{X}=\left[ X_1,\cdots,X_K \right] ^T$, $ i=1,\cdots, 2^K$ \\
$\varPhi_i$ & Enumerate all $2^K$ possibilities of $|\varPhi \rangle =|\sum_{k=1}^K{\phi _k}\rangle$, $ i=1,\cdots, 2^K$ \\
$\phi_{\max ,\mathrm{BS}}$  & The maximum effective received amplitude at the BS \\
$\phi_{\max ,k}$  & The maximum coherent-state amplitude of user~$k$\\
$\hat{D}\left( \boldsymbol{\gamma } \right)$  & Quantum displacement operator  \\
$\boldsymbol{W}_{\mathrm{sub}}$  & Sample weight set of Algorithm 2 \\
$\boldsymbol{\varPhi}_{\mathrm{sub}}$  & Sample set of Algorithm 2 \\
$\varXi \left( \boldsymbol{X} \right) $  & The proposal probability distribution of Algorithm 2 \\
$S$  & The sample size of Algorithm 2  \\
$\theta$ & The degrees of freedom of $t$-distribution\\
    \hline
  \end{tabular}
\end{table}
\section{System Model and Assumptions}
\label{sec:SysModel}

As illustrated in Fig.~\ref{fig:system}, we examine an uplink NOMA-OQC system with $K$ users. Let $ k=1,\cdots,K$ denote the indices of the users. Classical information from each user is mapped into $\boldsymbol{X}=\big[\ X_1,\cdots,X_K \big] ^T$, and transmitted via BPSK modulation using coherent states, where $X_k\in \mathbb{R}$ denotes the BPSK signal of the $k$-th user, and $\boldsymbol{X} \in \mathbb{R}$ collect the BPSK signals of all $K$ users. The coherent state emitted by a laser can be represented by a series of orthogonal bases $\{|m\rangle ,m=0,1,2,\cdots \}$ in an infinite-dimensional Hilbert space, as given by\cite{Glauber} 
\begin{equation}
\label{eq:alpha}
|\alpha\rangle =e^{-\frac{1}{2}|\alpha |^2}\sum_{m=0}^{\infty}{\frac{{\alpha }^{m}}{\sqrt{m!}}|m\rangle},
\end{equation}
    where $\alpha = |\alpha| e^{i\varphi } \in \mathbb{C}$ is a complex number, $|\alpha |^2$ is the average number of photons in the coherent state $|\alpha\rangle$, and $|m\rangle $ is the Fock state containing exactly $m$ photons. 
    If $ X_k=1$, then $\varphi=0$, and the coherent state emitted by user $k$ is $|\alpha _k\rangle =|\sqrt{\tau_k}\alpha \rangle$ with $\tau_k$ being the power allocation factor of user $k$ \cite{power1}. If $ X_k=-1$, then $\varphi=\pi$ and $|\alpha _k\rangle =|-\sqrt{\tau_k}\alpha \rangle$. 

Coherent states arrive at the receiver after undergoing atmospheric turbulence and are detected using a PNRD-Kennedy receiver. The photon loss during the detection process is modeled as a lossy photon channel \cite{lossy}. In the NOMA-based system supporting simultaneous transmissions of $K$ users, the receiver operates within a single mode of the Hilbert space for the reception of coherent states \cite{Multiaccess}. Consequently, individual signals from $K$ users undergo  coherent linear superposition, and can be physically realized using linear optical components such as beam splitters \cite{Diversity}. 
The coherent state at the PNRD-Kennedy receiver is $|\beta \rangle =\left|\sum_{k=1}^K{\sqrt{h_k}\alpha _k}\right\rangle $ \cite{Multiaccess,Diversity}. Here, $h_k=h_{k,t}h_{k,l}$ is the fading coefficient, $h_{k,t}$ is the transmissivity of the atmospheric turbulence, and $h_{k,l}$ the path loss.

With the PNRD-Kennedy receiver in Fig. 1, the received coherent state $|\beta\rangle$ is first displaced by a displacement operator $\hat{D}\big( \boldsymbol{\gamma } \big)$,  where $\boldsymbol{\gamma }$ is the displacement value. The displacement operation $\hat{D}\big( \boldsymbol{\gamma } \big)$ is achieved through a local oscillator (LO) and a beam splitter. Assuming a lossless beam splitter, then
\begin{equation}
\label{eq:displacement}
\hat{D}\big( \boldsymbol{\gamma } \big) |\beta \rangle =|\beta +\boldsymbol{\gamma }\rangle,
\end{equation}
where $\hat{D}\big( \boldsymbol{\gamma } \big) =e^{\big(\boldsymbol{\gamma }\hat{a}^{\dagger}-\boldsymbol{\gamma }^*\hat{a}\big)}$, with ${\hat{a}}^{\dagger}$ being the creation operator, $\hat{a}$ being the annihilation operator, and $\boldsymbol{\gamma }^*$ being the conjugate of $\boldsymbol{\gamma }$. 
Without loss of generality, we assume that $\boldsymbol{\gamma }=\sum_{k=1}^K{\gamma _k}$,  $\gamma _k=\sqrt{h_k\tau_k}\alpha$, and $\alpha$ is a real number \cite{Diversity}. 

Following the displacement operation, the quantum state is shifted from its original position $|\beta\rangle$ to $|\beta +\boldsymbol{\gamma }\rangle$ and can be detected by the PNRD.
For the brevity of notation and following the NOMA principle, we define the received coherent state as 
\begin{equation}
\label{eq:varPhi}
\begin{aligned}
|\varPhi \rangle =\Big|\sum_{k=1}^K{\phi _k}\Big\rangle =\Big|\sum_{k=1}^K{( \sqrt{h_k}\alpha _k+\gamma _k )}\Big\rangle,   
\end{aligned}
\end{equation}
where $|{\phi _k}\rangle =|{ \sqrt{h_k}\alpha _k+\gamma _k }\rangle$ denotes the received coherent state of user $k$. The physical interpretation of a coherent state $|{\phi _k}\rangle$ is the average number $\left| \phi_k \right|^2$ of photons, where $\phi_k$ represents the corresponding amplitude \cite{Quantum_Communication}. 

\begin{figure*}[htbp]
\centering
\includegraphics[width=0.75\textwidth]{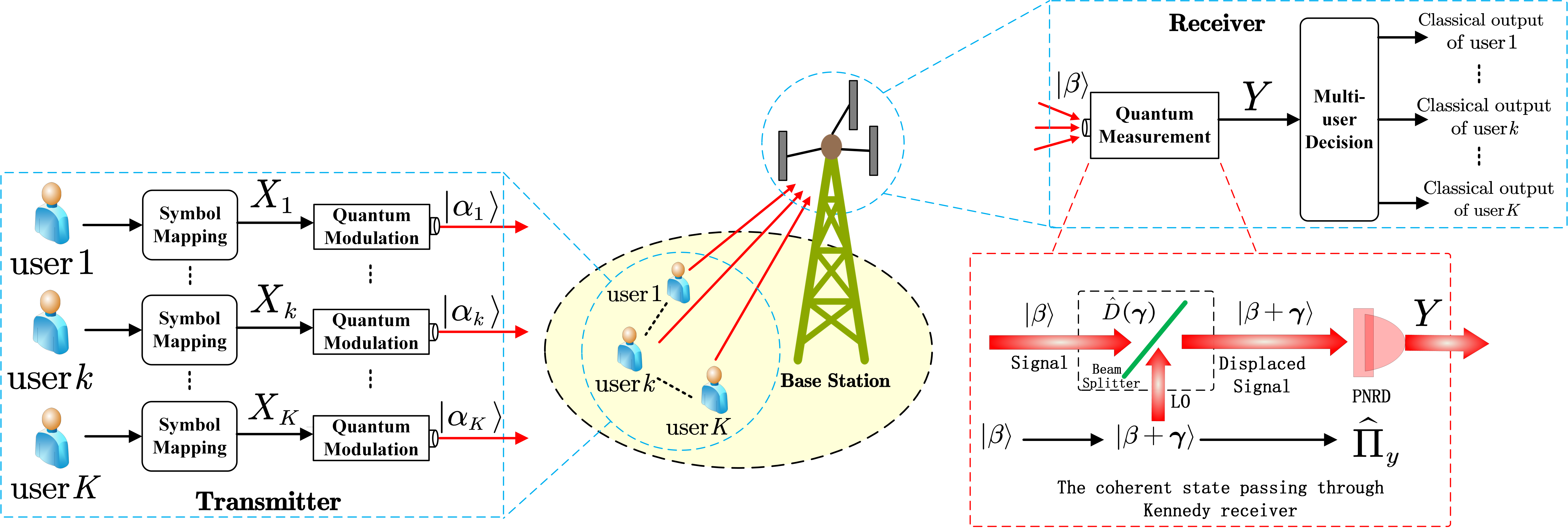}
\caption{ An illustration of the considered uplink NOMA-OQC system, where $K$ users form an uplink NOMA cluster, each user is equipped with a laser to emit coherent states, and the BS receives those coherent states through the PNRD-Kennedy receiver. Here, LO stands for local oscillator.}
\label{fig:system}
\end{figure*}

The received coherent state is measured using a PNRD characterized by a set of positive-operator valued measure (POVM) operators $\big\{ \hat{\Pi}_y \big\} _{y=0}^{\infty}$ \cite{Threshold}. The POVM measurement operator $\hat{\Pi}_y$ for $y$-photon detection can be expressed as~\cite{TPSK}
\begin{equation}
\label{eq:POVM}
\begin{aligned}
\hat{\Pi}_y= \sum_{\varepsilon =y}^{+\infty}&\binom{\varepsilon}{y}
\eta ^{y}\big( 1-\eta \big) ^{\varepsilon -y}|\varepsilon \rangle \langle \varepsilon |,
\end{aligned}
\end{equation}
where  $\eta\in \big[ 0,1 \big] $ is the transmission coefficient of the lossy photon channel under consideration \cite{lossy}.
When the background noise is negligible, the probability of detecting $y$ photons in the coherent state $|\varPhi \rangle $ is~\cite{lossy,TPSK} 
\begin{equation}
    \label{eq:basePr_group_ideal}
    \begin{aligned}
    \mathrm{Pr}\big( y\mid\varPhi \big)=\langle \varPhi |\hat{\Pi}_y|\varPhi\rangle=\frac{\big( \eta|\varPhi|^2\big) ^y}{y!}e^{-\eta|\varPhi |^2},  
\end{aligned}
\end{equation} 
where $ |\varPhi|^2$ is the average number of photons in $|\varPhi \rangle $. This measurement operation collapses the quantum state $|\varPhi \rangle $, yielding a definite, classical outcome $y$.

This study explicitly captures the impact of the background noise at the quantum receiver, including environmental radiation and the dark count of the PNRD, both of which follow Poisson distributions~\cite{backgroundnoise1}. The background noise can be represented as a coherent field $|B\rangle$, with the average of $\big| B \big|^2$ photons~\cite{backgroundnoise1,detection}. 
The average number of photons detected by the PNRD is increased to $ |\varPhi|^2+\left| B \right|^2$~\cite{TPSK,near-optimum,Nature2}. 

Define $n _b=\eta \big| B \big|^2$ to capture the influence of background noise. The conditional probability mass function (PMF) of detecting a total number $Y$ of received photons under the background noise obeys the following Poisson distribution:
\begin{subequations}
    \label{eq:basePr_group}
    \begin{flalign}
    P\big( \boldsymbol{X} \big) 
    &\triangleq \mathrm{Pr}\big( Y=y\mid X_1,\cdots,{ X_K} \big) \\ 
    &=\frac{\big( \eta|\varPhi|^2+n_b \big)^y}{y!}\cdot \exp \big[ - \big( \eta |\varPhi |^2+n_b \big) \big]. 
\end{flalign}
\end{subequations} 
SIC can be applied at the receiver to detect the coherent states from different users. Assume that the users are sorted in descending order of channel gain. 
Users 1 to $K$ are detected and canceled sequentially at the receiver.

\section{Sum-Rate Analysis of Uplink NOMA-OQC}
\label{sec:Analysis}
In this section, we analyze the sum-rate of the considered uplink NOMA-OQC systems. We start by deriving the exact expression for the sum-rate based on the definition of mutual information. Then, we derive both upper and lower bounds for the sum-rate, followed by an asymptotically tight, mathematically tractable approximation of the sum-rate. 

\subsection{Analysis of Explicit Sum-Rate} 

 Information transmission is achieved through the utilization of coherent states that can be received by the PNRD-Kennedy receiver. In the SIC process, the PNRD-Kennedy receiver uses POVM operators to measure quantum states, enabling precise photon-number resolution of received coherent states and facilitating sequential user signal decoding\cite{Access}. We employ the quantum mutual information between input and output states to characterize the system's transmission rate \cite{Cambridge,Quantum-mechanical}. The achievable rate for the user with the highest channel gain is given by
    \begin{equation}
    I_1(X_1; Y) = H(Y) - H(Y|X_1),
    \end{equation}
   where $H(\cdot )$ gives Shannon entropy, $H\left( \cdot |\cdot \right) $ stands for conditional entropy, and $I\left( \cdot \right) $ stands for mutual information. 

    When evaluating the \(k\)-th user,  the information of the first \((k-1)\) users is decoded and canceled. The conditional entropy \(H(Y|X_1, \cdots, X_{k-1})\) is used to evaluate the mutual information of the \(k\)-th user. The remaining signals from the other \((K-k)\) users are treated as interference, yielding the rate expression for the \(k\)-th user:
    \begin{equation}
    \begin{aligned}
    &I_k(X_k; Y|X_1, \cdots, X_{k-1}) \\
    &= H(Y|X_1, \cdots, X_{k-1}) - H(Y|X_1, \cdots, X_k).
    \end{aligned}
    \end{equation}

    Finally, we obtain the sum-rate for the uplink multi-user OQC system, as given by 
    \begin{subequations}
    \label{eq:SumRate_H}
    \begin{flalign}
    I_{\mathrm{sum}}&=I_1\big( X_1;Y\big)+\sum\limits_{k=2}^{K}{I_k\big( X_k;Y|X_1,\cdots,X_{k-1}\big)}\label{eq:SumRate,0}\\
    &=\underbrace{H\big(Y\big)-H\big(Y|X_1\big)}_{I_1\big( X_1;Y\big)}+\underbrace{H\big(Y|X_1\big)-H\big(Y|X_1,X_2\big)}_{I_2\big( X_2;Y|X_1\big)}+\notag\\
    &\!\cdots\!+\!\underbrace{H\big(Y|X_1,\!\cdots\!,X_{K-1}\big)\!-\!H\big(Y|X_1,\!\cdots\!,X_K\big)}_{I_K\big( X_K;Y|X_1,\cdots,X_{K-1}\big)}\label{eq:SumRate,1}\\
    &=H\big(Y\big)-H\big(Y|X_1,\cdots,X_K\big), \label{eq:SumRate,2}
    \end{flalign}
    \end{subequations}
    where \eqref{eq:SumRate,0} is based on the sum mutual information for all users. 
    \eqref{eq:SumRate,1} is based on the definition of mutual information and the SIC process. Simplifying \eqref{eq:SumRate,1}
yields \eqref{eq:SumRate,2}.

The photon number distribution $P\big( \boldsymbol{X} \big)$ of received coherent state $|\varPhi \rangle $, measured by PNRD, follows Poisson statistics, as given in \eqref{eq:basePr_group}. Consequently, we can employ Shannon entropy to measure the information of the coherent states in the form of their uncertainty \cite{Survey}. 
Substituting \eqref{eq:basePr_group} into \eqref{eq:SumRate,2} leads to a new expression for the sum-rate of the uplink NOMA-OQC system, as dictated in the following theorem.

\begin{theorem}
The expression for the sum-rate of the uplink NOMA-OQC system is given by
\begin{align}
\label{eq:SumRate}
I_{\mathrm{sum}}&=K\!\! +\!\! \frac{1}{2^K\ln 2}\!\!\sum_{y=0}^{+\infty} \bigg\{\!
\sum_{i=1}^{2^K} \Big[ \frac{(\eta\big| \varPhi_i \big|^2\!\!\!+\! n_b)^y}{y!} \exp( \!-\!\eta\big| \varPhi_i \big|^2\!\!- \!n_b ) \nonumber\\
&
\big.\times \ln \frac{\big( \eta\big| \varPhi_i \big|^2+ n_b \big)^y}{y!} \exp \big( -\eta\big| \varPhi_i \big|^2- n_b \big)\Big] \nonumber\\
& - \sum_{i=1}^{2^K} \frac{\big( \eta\big| \varPhi_i \big|^2+ n_b \big)^y}{y!} \exp \big( -\eta\big| \varPhi_i \big|^2- n_b \big) \nonumber\\
&\big.\cdot\ln \Big[ \sum_{i=1}^{2^K} \frac{\big( \eta\big| \varPhi_i \big|^2\!+ \!n_b \big)^y}{y!} \exp \big( \!-\!\eta\big| \varPhi_i \big|^2\!-\! n_b \big) \Big]
\bigg\}.   
\end{align}
\end{theorem}

\begin{proof}
    See \textbf{Appendix~A}.
\end{proof}

  The series in \eqref{eq:SumRate} converges, because \(P\left( \boldsymbol{X}_i \right)=\frac{(\eta |\varPhi_i|^2 + n_b)^y}{y!} \exp(-\eta |\varPhi_i|^2 - n_b)\) corresponds to the PMF of the Poisson distribution with finite parameter \( \eta |\varPhi_i|^2 + n_b\) and 
  satisfies 
  \(\sum_{y=0}^\infty P\left( \boldsymbol{X}_i \right) = 1\). As \(y\to\infty\), 
    \(P\left( \boldsymbol{X}_i \right) \to 0\) decays super-exponentially, ensuring the convergence of both \(\ln P\left( \boldsymbol{X}_i \right)\) and \(\ln\sum_{i=1}^{2^K} P\left( \boldsymbol{X}_i \right)\) and, in turn, the convergence of \eqref{eq:SumRate}.
  Particularly, $|\varPhi|^2$ and $n_b$ are experimentally measurable and remain stable over practical timescales \cite{detection}. They can be treated as deterministic ensemble averages, while the inherent Poisson statistics are characterized by $P\left( \boldsymbol{X_i} \right)$. \textbf{Theorem 1} preserves the statistical nature of coherent states while enabling sum-rate analysis and optimal power allocation. 
  
  The right-hand side (RHS) of \eqref{eq:SumRate} involves a complex expression with infinite series, making \eqref{eq:SumRate} a non-convex and non-concave function. Consequently, directly utilizing \eqref{eq:SumRate} for mathematical analysis and power allocation design can be cumbersome. For this reason, we next derive the upper and lower bounds, as well as an asymptotically tight, closed-form approximation for the sum-rate of the considered system.

\subsection{Upper and Lower Bounds of Sum-Rate}
We begin by converting the infinite series in \eqref{eq:SumRate} into infinite integrals. Then, we proceed with the deflation process and integrate the resulting expressions. With Jensen’s inequality and mathematical manipulations, we derive both upper and lower bounds for the sum-rate, as follows.

\begin{lemma}
An asymptotically tight lower bound for the received data rate at the quantum receiver
is given by
\begin{equation}
\label{Sum_rate_lowbound}
\begin{aligned}
I_{\mathrm{low}}^{\big( \mathrm{GA} \big)}=K+\frac{1}{2^K\ln 2}\big[ \varphi \big( \boldsymbol{\varPhi } \big) -f\big( \boldsymbol{\omega }_{\mathrm{low}} \big) \big], 
\end{aligned}
\end{equation}
where $\boldsymbol{\varPhi }=\big[ \varPhi_1,\cdots,\varPhi_{2^K} \big]$, $\boldsymbol{\omega }_{\mathrm{low}}=\big[ \omega_{\mathrm{low},1},\cdots,\omega_{\mathrm{low},2^K} \big]$, $\varphi \left( \boldsymbol{\varPhi } \right)$ is defined in closed-form, as given by 
\begin{equation}
\label{deqn_ex1_10}
 \varphi \left( \boldsymbol{\varPhi }\right)=  
- \frac{1}{2} \ln \Big[ \big( 2\pi \big) ^{2^K}n_b\prod\nolimits_{i=1}^{2^K}{( \eta\big| \varPhi_i \big|^2+n_b)} \Big] -2^{K-1},
\end{equation}
and $f(\boldsymbol{\omega}_{\mathrm{low}})$ is also in closed-form, as given by
\begin{equation}
    \label{eq:SumRate_low}
    \begin{aligned}
&f\big( \boldsymbol{\omega }_{\mathrm{low}} \big)= 
\ln \bigg\{ \sqrt{\frac{\pi}{2}}\Big[ \sum_{i=1}^{2^K}\sum_{j =1}^{2^K}\binom{2^K}{i} \binom{2^K}{j}\times
\big.\big.\\
&~~\big.\big.
\exp \big( \frac{\omega_{\mathrm{low},i}-\omega_{\mathrm{low},j}}{2} \big) ^2 \Big] \bigg\}  +\xi_{\mathrm{low}} -\frac{2^{K-1}\ln \left( 2\pi n_b \right)}{\sqrt{2n_b}},
\end{aligned}
\end{equation}
where $\xi _{\mathrm{low}}$ is a constant, and the variables in $\boldsymbol{\omega }_{\mathrm{low}}=\big[ \omega_{\mathrm{low},1},\cdots,\omega_{\mathrm{low},2^K} \big]$ are given by 
\begin{equation}
    \label{w_low}
\omega _{\mathrm{low},j}=\frac{\eta}{2\sqrt{2n_b}}\Big(\displaystyle\sum_{i=1}^{2^K}{\big|\varPhi_i \big|^2-2\left| \varPhi_j \right|^2}\Big), j=1,\cdots,2^K.
\end{equation}
\end{lemma}

\begin{proof}
    See \textbf{Appendix~B}.
\end{proof}

\begin{lemma}
An asymptotically tight upper bound of the received data rate, denoted as $I_{\mathrm{up}}^{\left( \mathrm{GA} \right)}$, is given by
\begin{equation}
\label{Sum_rate_upbound}
\begin{aligned}
I_{\mathrm{up}}^{\big( \mathrm{GA} \big)}=K+\frac{1}{2^K\ln 2}\big[ \varphi \big( \boldsymbol{\varPhi } \big) -f\big( \boldsymbol{\omega }_{\mathrm{up}} \big) \big], 
\end{aligned}
\end{equation}
where $\boldsymbol{\omega }_{\mathrm{up}}=\big[ \omega_{\mathrm{up},1},\cdots,\omega_{\mathrm{up},2^K} \big]$ and $f(\boldsymbol{\omega}_{\mathrm{up}})$ is given by 
\begin{align}
f& \big( \boldsymbol{\omega }_{\mathrm{up}} \big)=
  \nonumber\\
&\ln \bigg\{ \sqrt{\frac{\pi}{2}}\Big[ {\sum_{i=1}^{2^K}{\sum_{j =1}^{2^K}{{\binom{2^K}{i}}{\binom{2^K}{j}}\exp \big( \frac{{\omega _{\mathrm{up},i}}-{\omega _{\mathrm{up},j}}}{2} \big) ^2}}} \Big] \bigg\} \nonumber\\
&+\xi _{\mathrm{up}} -\frac{2^{K-1}\ln \big[ 2\pi \big( \eta | \sum_{i=1}^{K}{\phi_i}|^2+n_b \big) \big]}{\sqrt{2\big( \eta\big| \displaystyle\sum_{k=1}^{K}{\phi_k} \big|^2+n_b \big)}}.\label{eq:SumRate_up}
\end{align}
Here, $\xi _{\mathrm{up}}$ is a constant, and the variables in $\boldsymbol{\omega }_{\mathrm{up}}=\left[ \omega_{\mathrm{up},1},\cdots,\omega_{\mathrm{up},2^K} \right]$ are defined as 
\begin{equation}
    \label{w_up}
\omega _{\mathrm{up},j}=\frac{\eta\big(\displaystyle\sum_{i=1}^{2^K}{\left| \varPhi_i \right|^2-2\left| \varPhi_j \right|^2}\big)}{2\sqrt{2\big( \eta\big| \displaystyle\sum_{k=1}^{K}{\phi_k} \big|^2+n_b \big)}},\,j=1,\cdots,2^K.
\end{equation}

\end{lemma}

\begin{proof}
    See \textbf{Appendix~C}.
\end{proof}

\subsection{Asymptotically Tight Approximation of Sum-Rate}
The lower bound in \eqref{Sum_rate_lowbound} and the upper bound in \eqref{Sum_rate_upbound} are asymptotically tight, as $n_b \rightarrow +\infty$. By taking the average of the upper and lower bounds, a more precise and asymptotically accurate approximation of the sum-rate can be obtained:

\begin{theorem}
As $n_b \rightarrow +\infty$, both the lower bound $I_{\mathrm{low}}^{\big( \mathrm{GA} \big)}$ in \eqref{Sum_rate_lowbound} and the upper bound $I_{\mathrm{up}}^{\left( \mathrm{GA} \right)}$ in \eqref{Sum_rate_upbound} asymptotically converge to $K + \frac{\varphi \left( \boldsymbol{\varPhi } \right) - f\left( 0 \right)}{2^K \ln 2}$. Consequently, an asymptotically tight approximation of \eqref{eq:SumRate} is given by
\begin{equation}
\label{eq:Asymptotic_Sumrate}
\begin{aligned}
 &I_{\mathrm{sum}}\xrightarrow{n_b\rightarrow +\infty}\tilde{I}_{\mathrm{sum}}^{\big( \mathrm{GA} \big)}=\frac{I_{\mathrm{low}}^{\big( \mathrm{GA} \big)}+I_{\mathrm{up}}^{\left( \mathrm{GA} \right)}}{2}\,\,   \\
    &=K+\frac{1}{2^K\ln 2}\bigg[ \varphi \big( \boldsymbol{\varPhi } \big) -\frac{f\big( \boldsymbol{\omega }_{\mathrm{low}} \big) +f\big( \boldsymbol{\omega }_{\mathrm{up}} \big)}{2} \bigg], 
\end{aligned}
\end{equation}
where $I_{\mathrm{low}}^{\big( \mathrm{GA} \big)}<\tilde{I}_{\mathrm{sum}}^{\big( \mathrm{GA} \big)}<I_{\mathrm{up}}^{\big( \mathrm{GA} \big)}$.
\end{theorem}

\begin{proof}
    This readily follows from \textbf{Lemmas~1} and \textbf{2}.
\end{proof}
The asymptotic approximation of the sum-rate in \eqref{eq:Asymptotic_Sumrate} offers a mathematically tractable alternative to the exact sum-rate in \eqref{eq:SumRate}, since \eqref{eq:Asymptotic_Sumrate} only involves closed-form functions, including $\varphi \left( \boldsymbol{\varPhi } \right)$ in \eqref{deqn_ex1_10}, $f\left( \boldsymbol{\omega }_{\mathrm{low}} \right)$  in \eqref{eq:SumRate_low}, and $f\left( \boldsymbol{\omega }_{\mathrm{up}} \right)$ in~\eqref{eq:SumRate_up}. 
This approximation becomes more accurate as $n_b \rightarrow +\infty$. 
For practical systems, typically, $n_b>1$, e.g., in solar radiation-dominated environments \cite{solar_radiation}. This can introduce inaccuracy of less than 0.6 dB in system performance \cite[ Fig.~2]{Yongkang}. Such inaccuracy is significant when all users transmit signal ``$-1$''. As the users increase, the probability that all users transmit signal ``-1'' decreases, making this inaccuracy negligible.
 
This analysis offers engineering insights to simplify system design, enable efficient power allocation, 
and ensure robustness in high-noise environments. 
    Specifically, the asymptotically tight sum-rate approximations established in \textbf{Lemmas 1} and \textbf{2} allow for computationally efficient evaluation of system performance and facilitate power allocation optimization in uplink NOMA-OQC systems. Moreover, the bounds can assist in optimizing critical system parameters, such as displacement, noise, and photon loss, which are essential for enhancing quantum communication receiver performance.

\section{Proposed Coherent-State Power Allocation}
\label{sec:Algorithm}
In this section, the problem of interest is formulated. Utilizing variable substitution and SCA, we transform the problem into a convex optimization problem. Finally, we develop a new power allocation algorithm suitable for scenarios involving a small-to-medium number of users.

\subsection{SCA-Based Coherent-State Power Allocation}
Our goal is to determine the allocation of coherent-state power that maximizes the sum-rate of the considered system:
\begin{subequations}
\begin{align}
\mathbf{P1:} & \max_{|\phi_1|^2 ,\cdots, |\phi_K|^2} \quad \tilde{I}_{\mathrm{sum}}^{\left( \mathrm{GA} \right)} \label{eq:OriginalProblem,obj} \\
&\quad\mathrm{s.t.} \quad |\sum_{k=1}^K \phi_k|^2 \leqslant |\phi_{\max,\mathrm{BS}}|^2, \label{eq:OriginalProblem,TotalCst} \\
&\quad\quad\quad\quad|\phi_k|^2 \leqslant |\phi_{\max,k}|^2, \,\forall k=1,\cdots,K.\label{eq:OriginalProblem,ZeroCst}
\end{align}
\end{subequations}

As mentioned earlier in Section~\ref{sec:SysModel}, the physical interpretation of a coherent state $|{\phi _k}\rangle$ is the average number of photons $\left| \phi_k \right|^2$, where $\phi_k$ represents the corresponding amplitude \cite{Quantum_Communication}. For this reason, the optimization variables of Problem \textbf{P1} are $|\phi_i|^2$, $i=1,\cdots,K$. Considering the saturated received signals of PNRD \cite{Yongkang2}, constraint \eqref{eq:OriginalProblem,TotalCst} specifies the maximum effective receiving power $|{\phi_{\max,\mathrm{BS}}}|^2$ at the BS. Additionally, $|{\phi_{\max,k}}|^2$ denotes the maximum coherent-state power of the $k$-th user, subject to constraint \eqref{eq:OriginalProblem,ZeroCst}. 

In the objective function \eqref{eq:Asymptotic_Sumrate} of Problem \textbf{P1}, $\varphi \left( \boldsymbol{\varPhi } \right)$ is a convex function. The term $f(\boldsymbol{\omega}_{\mathrm{low}})$ represents the logarithm of the exponential sums plus a constant term, thereby yielding a convex function. On the other hand, $f(\boldsymbol{\omega}_{\mathrm{up}})$ encompasses the logarithm of the sum of exponential functions and a fractional expression, i.e., 
$-\frac{2^{K-1}\ln [2\pi (\eta\big| \sum_{k=1}^{K}{\phi_k} \big|^2+n_b)]}{\sqrt{2 (\eta\big| \sum_{k=1}^{K}{\phi_k} \big|^2+n_b)}}$, which renders $f(\boldsymbol{\omega}_{\mathrm{up}})$ non-convex and non-concave. As a consequence, $I_{\mathrm{sum}}^{\big( \mathrm{GA} \big)}$ exhibits non-convexity, and can be non-trivial to solve.

We resort to employing variable substitution and SCA to convexify the objective function, i.e., identifying a surrogate concave function for the objective function. Consequently, we convert the Problem \textbf{P1} into a convex optimization problem, since both constraints \eqref{eq:OriginalProblem,TotalCst} and \eqref{eq:OriginalProblem,ZeroCst} are linear and hence affine. For conciseness, with variable substitution, we define 
\begin{equation}
\label{mu_t}
\begin{aligned}
\mu =\big( 2\pi \big) ^{2^K}n_b\prod\nolimits_{i=1}^{2^K}{\big( \eta\big| \varPhi_i \big|^2+n_b \big)}.
\end{aligned}
\end{equation}
Then, $\varphi \left( \boldsymbol{\varPhi } \right)$ in \eqref{deqn_ex1_10} can be expressed as
\begin{equation}
\label{varphi_mu}
\begin{aligned}
\varphi \left( \mu \right) =- \frac{1}{2} \ln \left( \mu \right) -2^{K-1},
\end{aligned}
\end{equation}
which exhibits convexity with respect to $\mu$. 

By applying the first-order Taylor series approximation to expand the function $\varphi \left( \mu \right) $ at the SCA local point 
\begin{equation}
\label{mu_t_1}
\begin{aligned}
\mu ^{(t-1)} = (2\pi)^{2^K}n_b\prod\nolimits_{i=1}^{2^K}{\big( \eta\big|\varPhi_i^{(t-1)}\big|^2+n_b \big)},
\end{aligned}
\end{equation}
we derive the surrogate concave function of $\varphi \left( \mu \right) $ for the $t$-th SCA iteration, as given by  
\begin{equation}
\label{SCA_1}
\begin{aligned}
&\varphi _{\mathrm{SCA}}^{( t )}\big( \mu ,\mu ^{( t-1 )} \big) \\
&=\varphi _{\mathrm{SCA}}^{( t )}\big( \mu ^{( t-1 )} \big) +{\varphi _{\mathrm{SCA}}^{( t )}}^{\prime}\big( \mu ^{( t-1 )} \big) \big( \mu -\mu ^{( t-1 )} \big) \\        
&=-\frac{1}{2} \ln \big( \mu ^{( t-1)} \big) -2^{K-1}-\frac{1}{2\mu ^{( t-1)}}\big( \mu -\mu ^{( t-1)} \big),
\end{aligned}
\end{equation}
where the superscript ``$^{\big( t \big)}$'' indicate the $t$-th SCA iteration. 

For the function $f\left( \boldsymbol{\omega }_{\mathrm{up}} \right)$, we also perform variable substitution by defining 
\begin{equation}
\label{upsilon}
\begin{aligned}
\upsilon =2\pi \big( \eta\big| \displaystyle\sum_{k=1}^{K}{\phi_k} \big|^2+n_b \big),
\end{aligned}
\end{equation}
with
\begin{equation}
\label{h_upsilon}
\begin{aligned}
h\left( \upsilon \right) \triangleq -\frac{2^{K-1}\sqrt{\pi}\ln \upsilon}{\sqrt{\upsilon}}.
\end{aligned}
\end{equation} 
Then, the fractional expression 
$-\frac{2^{K-1}\ln [2\pi (\eta\left| \sum_{k=1}^{K}{\phi_k} \right|^2+n_b)]}{\sqrt{2 (\eta\left| \sum_{k=1}^{K}{\phi_k} \right|^2+n_b)}}$ in $f\left( \boldsymbol{\omega }_{\mathrm{up}} \right)$ can be replaced by $h\left( \upsilon \right) $. By taking the second derivative of $h\left( \upsilon \right)$ with respect to~$\upsilon$, its concavity can be confirmed within the feasible region. 

Expanding the function $h\left( \upsilon \right)$ at the SCA local point 
\begin{equation}
\label{upsilon_t_1}
\begin{aligned}
\upsilon ^{( t-1 )}=2\pi \big( \eta\big| \displaystyle\sum_{k=1}^{K}{\phi_k^{( t-1)}}  \big|^2+n_b \big),
\end{aligned}
\end{equation}
to the first-order Taylor series, we obtain the surrogate convex function of $f\left( \boldsymbol{\omega }_{\mathrm{up}} \right)$ in the $t$-th SCA iteration, as given by
\begin{equation}
\label{eq:SumRate_up_SCA}
\begin{aligned}
f_{\mathrm{SCA}}^{( t )}\big( \boldsymbol{\omega }_{\mathrm{up}},\upsilon ,\upsilon ^{( t-1 )} \big)=\ln \bigg\{ \sqrt{\frac{\pi}{2}}\Big[ \sum_{i=1}^{2^K}\sum_{j =1}^{2^K}{\binom{2^K}{i}}{\binom{2^K}{j}}\big.\big.&\\
\big.\big.
\exp \big( \frac{\omega _{\mathrm{up},i}-\omega _{\mathrm{up},j}}{2} \big)^2 \Big] \bigg\} +\xi _{\mathrm{up}}+h_{\mathrm{SCA}}^{\left( t \right)}\big( \upsilon ,\upsilon ^{( t-1 )} \big) ,&
\end{aligned}
\end{equation}
where 
    \begin{align}\label{SCA_2}
&h_{\mathrm{SCA}}^{( t )}( \upsilon ,\upsilon ^{( t-1)} ) 
=h_{\mathrm{SCA}}^{( t )}( \upsilon ^{( t-1 )} ) \!+\!{h_{\mathrm{SCA}}^{( t )}}^{\prime}( \upsilon ^{( t-1 )} ) ( \upsilon \!-\!\upsilon ^{( t\!-\!1)} ) \nonumber\\
&=-\frac{2^{K-1}\sqrt{\pi}\ln( \upsilon ^{( t-1)})}{ \sqrt{\upsilon ^{( t\!-\!1 )}} }\!-\!
2^{K-2}\sqrt{\pi}\frac{1\!-\!\ln \upsilon ^{( t\!-\!1 )}}{\big[ \upsilon ^{( t\!-\!1 )} \big]^{\frac{3}{2}}}( \upsilon\!-\!\upsilon ^{( t\!-\!1 )} ). 
\end{align}
Here, $h_{\mathrm{SCA}}^{\left( t \right)}\left( \upsilon ,\upsilon ^{\left( t-1 \right)} \right)$ is the surrogate convex function of $h\left( \upsilon \right)$ for the $t$-th SCA iteration.
The convexity of $f_{\mathrm{SCA}}^{\left( t \right)}\left( \boldsymbol{\omega }_{\mathrm{up}},\upsilon ,\upsilon ^{\left( t-1 \right)} \right) $ can be confirmed since it encompasses the logarithm of the sum of exponential functions, followed by the addition with a constant term and a convex function.

By substituting \eqref{SCA_1} and \eqref {eq:SumRate_up_SCA} into \eqref{eq:Asymptotic_Sumrate}, the surrogate concave function of $I_{\mathrm{sum}}^{\left( \mathrm{GA} \right)}$ can be obtained as
\begin{equation}
\label{eq:SumRate_SCA}
\begin{aligned}
\tilde{I}_{\mathrm{sum},\mathrm{SCA}}^{( \mathrm{GA} ) ,( t )}=K+\frac{1}{2^K\ln 2}\cdot\big[\varphi_{\mathrm{SCA}}^{( t )}\big( \mu ,\mu ^{( t-1 )} \big)\big. &\\
 \big.-\frac{f\big( \boldsymbol{\omega }_{\mathrm{low}} \big) +f_{\mathrm{SCA}}\big( \boldsymbol{\omega }_{\mathrm{up}},\upsilon ,\upsilon ^{( t-1 )} \big)}{2} \bigg] ,&
\end{aligned}
\end{equation}
where the surrogate function is derived by subtracting the convex function $f\left( \boldsymbol{\omega }_{\mathrm{low}} \right)$ from the concave function $\varphi_ {\mathrm{SCA}}^{\left( t \right)}\left( \mu ,\mu ^{\left( t-1 \right)} \right)$ and subsequently subtracting the convex function $f_{\mathrm{SCA}}^{\left( t \right)}\left( \boldsymbol{\omega }_{\mathrm{up}},\upsilon ,\upsilon ^{\left( t-1 \right)} \right)$, followed by the operation with constant terms. Clearly, the surrogate function in \eqref{eq:SumRate_SCA} is a concave function. 

As a result, Problem \textbf{P1} is eventually reformulated as  
\begin{subequations}
    \begin{align}
        \mathbf{P2}: & \max_{|\phi_1|^2, \cdots, |\phi_K|^2, \mu, \upsilon } \tilde{I}_{\mathrm{sum},\mathrm{SCA}}^{\big( \mathrm{GA} \big) ,\big( t \big)} \label{constraint_I}\\
        & \quad\quad\mathrm{s.t. } 
        \quad \mu \geqslant \big( 2\pi \big) ^{2^K}n_b\prod\nolimits_{i=1}^{2^K}{\big( \eta\big| \varPhi_i\big|^2+n_b \big)}, \label{constraint_u}\\
        & \quad\quad\quad\quad\upsilon \geqslant 2\pi \big( \eta\big| \displaystyle\sum_{k=1}^{K}{\phi_k} \big|^2+n_b \big), \label{constraint_v}\\
        & \quad\quad\quad\eqref{eq:OriginalProblem,TotalCst},\eqref{eq:OriginalProblem,ZeroCst}.\notag
    \end{align}
\end{subequations}

Constraints \eqref{constraint_u} and \eqref{constraint_v} arise from variable substitution during convexification. As the optimization problem involves maximizing a concave function over a closed convex interval specified by constraints \eqref{eq:OriginalProblem,TotalCst}, \eqref{eq:OriginalProblem,ZeroCst}, \eqref{constraint_u} and \eqref{constraint_v}, it becomes a convex optimization problem and can be solved effectively using the CVX toolbox. The solution obtained from the CVX toolbox in the $t$ iteration of the SCA serves as the local point for the $(t+1)$-th iteration. This updated local point is utilized to formulate the convex approximation Problem~\textbf{P2} for the $(t+1)$-th SCA iteration. Upon the convergence of the SCA iterations, the power allocation solution is generated.

\begin{algorithm}[t]
\caption{Power allocation for uplink NOMA-OQC}\label{alg:alg1}
\begin{algorithmic}
\item \hspace{0.35cm} \textbf{1: Input:}$~n_b,|\phi_\mathrm{\max ,BS}|^2,\eta, h_k, T_{\max}, \varepsilon _{\mathrm{SCA}},$
\item \hspace{1.8cm} $|\phi_{\max ,k}|^2,\forall,k=1\cdots,K;$
\item \hspace{0.35cm} \textbf{2: Initialization:}$~t=0;~|{\phi_k^{\left( 0\right)}}|^2,\forall k=1\cdots,K, $
\item \hspace{0.5cm}$  \textbf{3: while}~{\tilde{I}_{\mathrm{sum},\mathrm{SCA}}^{\left( \mathrm{GA} \right) ,\left( t \right)}-\tilde{I}_{\mathrm{sum},\mathrm{SCA}}^{\left( \mathrm{GA} \right) ,\left( t-1 \right)} }>\varepsilon _{\mathrm{SCA}}$  or 
\item \hspace{1.5cm}$ ~~~~t<T_{\max};$ \textbf{do}
\item \hspace{0.5cm}{\textbf{4: }\text{~~~Solving Problem \textbf{P2} using the CVX toolbox;} 
\item \hspace{0.5cm}$ \textbf{5: ~~~$t=t+1;$}$
\item \hspace{0.5cm}$ \textbf{6: end while;}$
\item \hspace{0.35cm} \textbf{7: Output:}$~|{\phi_k}|^2 =|{\phi_k^{\left( t\right)}}|^2,\forall k=1\cdots,K;$
}
\end{algorithmic}
\end{algorithm}

\subsection{Convergence and Complexity Analysis}
\label{Convergence Analysis: Algorithm 1}
\textbf{Algorithm \ref{alg:alg1}} summarizes the proposed coherent-state power allocation method, which iteratively solves the convex optimization Problem~\textbf{P2} from Step 3 to Step 6. In particular, Step 4 resolves Problem~\textbf{P2} using the CVX toolbox and then refreshes the local point for the next iteration of the SCA. 

\subsubsection{Convergence Analysis}

According to \cite{kkt1} and\cite{ kkt2}, \textbf{Algorithm \ref{alg:alg1}} converge to a solution that satisfies the KKT conditions of Problem \textbf{P2}. Specifically, the surrogate concave function $\varphi _{\mathrm{SCA}}^{\left( t \right)}\left( \mu ,\mu ^{\left( t-1 \right)} \right) $ in \eqref{SCA_1} and the corresponding original function $\varphi \left( \mu \right) $ in \eqref{varphi_mu} take the same value, which is $-\left( \frac{1}{2} \right) \ln \left( \mu ^{\left( t-1 \right)} \right) -2^{K-1}$, and the same gradient, which is $-\frac{1}{2\mu ^{\left( t-1 \right)}}$, at the local point $\mu =\mu ^{\left( t-1 \right)}$. Within the feasible domain, the condition $\varphi \left( \mu \right) \geqslant \varphi _{\mathrm{SCA}}^{\left( t \right)}\left( \mu ,\mu ^{\left( t-1 \right)} \right) $ always holds.
Similarly, the surrogate convex function $h_{\mathrm{SCA}}^{\left( t \right)}\left( \upsilon ,\upsilon^{\left( t-1 \right)} \right) $ in \eqref{SCA_2} and the corresponding original function $h\left( \upsilon \right) $ in \eqref{h_upsilon} share the same value, which is $-\frac{2^{K-1}\sqrt{\pi}\ln \left( \upsilon^{\left( t-1 \right)} \right)}{\left( \sqrt{\upsilon^{\left( t-1 \right)}} \right)}$, and the same gradient, which is $-2^{K-2}\sqrt{\pi}\frac{1-\ln \upsilon^{\left( t-1 \right)}}{\left[ \upsilon^{\left( t-1 \right)} \right] ^{\left( 3/2 \right)}}$, at the local point $\upsilon =\upsilon^{\left( t-1 \right)}$. Within the feasible domain, the condition $h\left( \upsilon \right) \geqslant h_{\mathrm{SCA}}^{\left( t \right)}\left( \upsilon,\upsilon^{\left( t-1 \right)} \right) $ is satisfied.

On the other hand, Problem~\textbf{P2} is a convex optimization problem and satisfies Slater’s condition at the $t$-th iteration of the SCA. In other words, there exists a point in the feasible domain that strictly satisfies the constraints of Problem~\textbf{P2}. The solution achieved by \textbf{Algorithm \ref{alg:alg1}} is in the interior of the feasible solution region of Problem~\textbf{P2}. 
According to \cite[Thm 1]{kkt1}, we can confirm that \textbf{Algorithm~\ref{alg:alg1}} converges to a KKT stationary point that is also a local optimum of original Problem \textbf{P1}. Since the solution lies in the interior of the feasible region, it constitutes a local optimum of Problem \textbf{P1}.

\subsubsection{Complexity Analysis}
\textbf{Algorithm \ref{alg:alg1}} solves Problem~\textbf{P2} using the interior point method during each SCA iteration. The complexity of using the interior point method to solve a convex problem is $\mathcal{O} \left( \max \left\{ N_c,N_v \right\} ^4\sqrt{N_v} \right) $, where $N_v$ and $N_c$ stand for numbers of variables and constraints, respectively~\cite{opt_variables,Convex}. For Problem~\textbf{P2}, $N_v=K+2$, $N_c=K+3$, and the complexity is  $\mathcal{O}\left(K^{4.5}\right)$ for the interior point method. On the other hand, the objective function \eqref{constraint_I} and the constraint  \eqref{constraint_u} involve a product of $2^K$ terms, $\left( \eta\left| \varPhi_i\right|^2+n_b \right)$, $\forall i$, to evaluate $\mu$ using \eqref{mu_t} per SCA iteration, incurring a complexity of $\mathcal{O}\left(2^K \right)$. Suppose the convergence accuracy of each SCA iteration is $\varepsilon_{\mathrm{SCA}}\in (0,1)$. The overall complexity of \textbf{Algorithm \ref{alg:alg1}} is $\mathcal{O}\left(2^K\cdot K^{4.5} \log\left(1/{\varepsilon_{\mathrm{SCA}}}\right)\right)$, which grows exponentially as the number of users $K$ increases.

\section{Low-Complexity  Power Allocation With Adaptive Importance Sampling}
\label{sec:Low-Complexity}

    This section addresses the exponential complexity of \textbf{Algorithm~\ref{alg:alg1}} in the presence of a medium-to-large number of users by proposing a low-complexity, approximate alternative that uses adaptive importance sampling to select \( S \) representative samples from the \( 2^K \)  possibilities (\( S \ll 2^K \))~\cite{Importance1}. 
    Adaptive importance sampling involves a proposal distribution, which encompasses the distribution of samples, and a target distribution to be approximated by the proposal distribution. We compute the weights of samples in the proposal distribution and iteratively update the parameters of the sampling process until a satisfactory approximation is achieved.

\subsection{Low-Complexity Solution with Importance Sampling}
\label{sec:Algorithm_A}

We initially identify the target distribution based on an asymptotically tight approximation of the sum-rate in \eqref{eq:Asymptotic_Sumrate}, where $\varphi(\varphi)$, $f(\boldsymbol{\omega}_{\mathrm{up}})$, and $f(\boldsymbol{\omega}_{\mathrm{low}})$ are all derived from $P(\boldsymbol{X})$ of \eqref{eq:basePr_group}. We approximate $P(\boldsymbol{X})$  using the Gaussian approximation PDF, denoted by $P^{\left( \mathrm{GA} \right)}(\boldsymbol{X})$, in the process of deriving the asymptotically tight approximation of the sum-rate \cite{Yongkang}; see Appendix B for details. 

The complexity of \textbf{Algorithm~\ref{alg:alg1}} arises from the computation process of $P^{\left( \mathrm{GA} \right)}(\boldsymbol{X}_i)$, which involves $2^K$ values.
We model the target distribution using a subset of $P^{\left( \mathrm{GA} \right)}\left( \boldsymbol{X}_i \right)$, $i=1,2,\cdots,2^K$. 

Next, we determine the proposal distribution to ensure comprehensive coverage of the significant probability regions of the target Gaussian distribution. We employ the $t$-distribution as the proposal distribution. This choice is effective for Gaussian targets, as the $t$-distribution's flexible shape parameters provide better coverage of the Gaussian's probability density while maintaining heavier tails to capture outlier regions \cite{Importance1,t_distribution}. The $t$-distribution is given as \cite{t_distribution}
\begin{equation}
\label{eq:t-distribution}
\begin{aligned}
\varXi \big( \boldsymbol{X}_i \big) =\frac{\varGamma \big( \frac{\theta +1}{2} \big)}{\sqrt{\theta \pi}\varGamma \big( \frac{\theta}{2} \big)}\big(1+\frac{{\boldsymbol{X}_i}^2}{\theta} \big)^{-\frac{\boldsymbol{X}_i+1}{2}},\forall i\in[1,2^K],
\end{aligned}
\end{equation}
where $\theta$ specifies the degrees of freedom, and $\varGamma \left( \cdot \right) $ stands for the Gamma function. 

Suppose that $S$ samples, $P^{\left( \mathrm{GA} \right)}\left( \boldsymbol{X}_{\pi(s)} \right)$, $s=1,\cdots,S$, are drawn from the set of $2^K$ values of the Gaussian approximation PDF, $P^{\left( \mathrm{GA} \right)}\left( \boldsymbol{X}_i \right)$, $ i = 1,\cdots,2^K $. 
$\pi(s)$ is the index of the $s$-th selected sample. $\pi(s)\neq \pi(s')$, if $s\neq s'$ and $1\leq s,s'\leq S$. $ 1\leq\pi(s),\forall s \leq 2^K$.
Within the $t$-distribution in \eqref{eq:t-distribution}, each sample is associated with an importance weight:
\begin{equation}
\label{weight}
\begin{aligned}
W_{\mathrm{sub},s}=\frac{P^{( \mathrm{GA} )}\big( \boldsymbol{X}_{\pi(s)} \big)}{\varXi \big( \boldsymbol{X}_{\pi(s)} \big)}, s=  1,\cdots,S.
\end{aligned}
\end{equation}
Given the proposal
distribution $\varXi \left( \boldsymbol{X}_{\pi(s)} \right)$, the importance weight $W_{\mathrm{sub},s}$ indicates the importance of each sample $\boldsymbol{X}_{\pi(s)}$ for the approximated target function $P^{\left( \mathrm{GA} \right)}\left( \boldsymbol{X}_{\pi(s)} \right)$.

The importance weights $W_{\mathrm{sub},s}$ may exhibit numerical instability when the proposal distribution $\varXi(\cdot)$ underestimates the target density.
We therefore normalize the weights:  
\begin{equation}
\label{weight_nom}
\begin{aligned}
\hat{W}_{\mathrm{sub},s}=\frac{W_{\mathrm{sub},s}}{\sum_{s=1}^S{W_{\mathrm{sub},s}}}, s= 1,\cdots,S.
\end{aligned}
\end{equation}
This process ensures a normalized distribution of importance across all samples.

To update the parameters of the iterative sampling process, including the sample set $\boldsymbol{\varPhi }_{\mathrm{sub}}$, the weight set $\boldsymbol{W}_{\mathrm{sub}}$, and the proposal distribution $\varXi \left( \boldsymbol{X}_{\pi(s)}\right)$, we compute the variation~$r$ in the weights between the current and previous iterations using $r=\sum_{s=1}^S{\left| W_{\mathrm{sub},s}^{\left( z\right)}-W_{\mathrm{sub},s}^{\left( z-1 \right)} \right|}$, where $z$ represents the number of iterations in the adaptive importance sampling process. If $r$ falls below the threshold $r_{\rm thre}$, the sampling process converges, resulting in the set of samples, $\boldsymbol{\varPhi }_{\mathrm{sub}}=\left[ \boldsymbol{\phi}_{\mathrm{sub},1},\cdots,\boldsymbol{\phi}_{\mathrm{sub},S} \right] $, along with their corresponding weights $\boldsymbol{W}_{\mathrm{sub}}=\left[ W_{\mathrm{sub},1},\cdots,W_{\mathrm{sub},S} \right] $. Based on the sample set and the weight set, we obtain the sampling result $\varPhi_{\mathrm{sub},s}= \boldsymbol{\phi}_{\mathrm{sub},s}\cdot W_{\mathrm{sub},s}$, $\forall s= 1,\cdots,S $.
 
We can substitute the final sample set $\left[\varPhi_{\mathrm{sub},1},\cdots,\varPhi_{\mathrm{sub},S} \right] $ into \eqref{eq:SumRate_SCA}. Then, the surrogate concave function of $I_{\mathrm{sum}}^{\left( \mathrm{GA} \right)}$ in \eqref{eq:SumRate_SCA} can be rewritten as
\begin{align}\label{eq:SumRate_sub}
&\tilde{I}_{\mathrm{sum,SCA},\mathrm{sub}}^{( \mathrm{GA} ),( t )}=K+\frac{1}{2^K\ln 2}\cdot \big[\varphi _{\mathrm{SCA}}^{( t )}\big( \mu _{\mathrm{sub}},\mu _{\mathrm{sub}}^{( t-1 )} \big) -\big.\nonumber \\	
&\big.\quad\frac{f_{\mathrm{sub}}\big( \boldsymbol{\omega }_{\mathrm{low},\mathrm{sub}} \big) +f_{\mathrm{SCA},\mathrm{sub}}\big( \boldsymbol{\omega }_{\mathrm{up},\mathrm{sub}},{\upsilon},{\upsilon^{( t-1 )}} \big)}{2}\big], 
\end{align}
where $f_{\mathrm{sub}}\left( \boldsymbol{\omega }_{\mathrm{low,sub}} \right)$ and $f_{\mathrm{SCA},\mathrm{sub}}\left( \boldsymbol{\omega }_{\mathrm{up,sub}},{\upsilon,}\upsilon^{\left( t-1 \right)} \right)$ can be obtained by 
replacing 
$\left[ \varPhi_1,\cdots,\varPhi_{2^K} \right]$ with $\left[ \varPhi_{\mathrm{sub},1},\cdots,\varPhi_{\mathrm{sub},S} \right]$ in \eqref{w_low} and \eqref{w_up}, and then substituting the resulting $\boldsymbol{\omega}_{\mathrm{low,sub}}$ and $\boldsymbol{\omega}_{\mathrm{up,sub}}$ into \eqref{eq:SumRate_low} and \eqref{eq:SumRate_up_SCA} with the summation bounds updated accordingly. 
Moreover,
\begin{equation}
\begin{aligned}
\mu _{\mathrm{sub}} =\big( 2\pi \big) ^{S}n_b\prod\nolimits_{s=1}^S{\big( \eta\big| \varPhi _{\mathrm{sub,}s} \big|^2+n_b \big)};
\end{aligned}
\end{equation}
\begin{equation}
\begin{aligned}
\mu _{\mathrm{sub}}^{(t-1)} =\big( 2\pi \big) ^{S}n_b\prod\nolimits_{s=1}^S{\big( \eta\big| \varPhi _{\mathrm{sub},s}^{(t-1)} \big|^2+n_b \big)}.
\end{aligned}
\end{equation}
As a consequence, Problem \textbf{P2} can be further rewritten as
\begin{subequations}
    \begin{align}
        \mathbf{P3}: & \max_{|\phi_1|^2, \cdots, |\phi_K|^2, \upsilon,\mu_{\mathrm{sub}},}\tilde{I}_{\mathrm{sum},\mathrm{SCA},\mathrm{sub}}^{( \mathrm{GA} ) ,( t )}\label{objective_u3}
        \\
        & \mathrm{s.t.} 
        \quad\mu_{\mathrm{sub}} \geqslant \big( 2\pi \big) ^{S}n_b\prod\nolimits_{s=1}^S{\big( \eta\left| \varPhi_{\mathrm{sub},s} \right|^2+n_b \big)}. \label{constraint_u3}       \\
        &\quad\quad\quad \eqref{eq:OriginalProblem,TotalCst}, \eqref{eq:OriginalProblem,ZeroCst},\eqref{constraint_v},\notag
    \end{align}
\end{subequations}
Problem \textbf{P3} maximizes a concave function over a closed convex interval specified by constraints \eqref{eq:OriginalProblem,TotalCst}, 
 \eqref{eq:OriginalProblem,ZeroCst}, \eqref{constraint_v} and \eqref{constraint_u3}. As can be analyzed in the same way as Problem \textbf{P2}, Problem \textbf{P3} is a convex optimization problem. 

\subsection{Convergence and Complexity Analysis}
\begin{algorithm}[t]
\caption{Low-Complexity Algorithm With Adaptive Importance Sampling}\label{alg:low-complexity}
\begin{algorithmic} 
\item {\textbf{1: Input:}$~n_b,|\phi_\mathrm{\max ,BS}|^2,\eta, h_k, T_{\max},\varepsilon _{\mathrm{SCA}},r_{\rm thre},$}
\item \hspace{1.4cm} $|\phi_{\max ,k}|^2,\forall k=1\cdots,K;$
\item {\textbf{2: Initialization:}}
\item \hspace{0.35cm} $~|{\phi_k^{\left( 0\right)}}|^2,\forall k=1\cdots,K;~\varXi \left( \boldsymbol{X}_{\pi(s)} \right), \forall s=1\cdots,S;$
\item \hspace{0.5cm}$ \boldsymbol{W}_{\mathrm{sub}};~\boldsymbol{\varPhi }_{\mathrm{sub}};~r;~z=0;~t=0$;
\item $\textbf{3: while }{\tilde{I}_{\mathrm{sum},\mathrm{SCA},\mathrm{sub}}^{\left( \mathrm{GA} \right) ,\left( t \right)}-\tilde{I}_{\mathrm{sum},\mathrm{SCA},\mathrm{sub}}^{\left( \mathrm{GA} \right) ,\left( t-1\right)} }>\varepsilon _{\mathrm{SCA}}$ or
\item \hspace{1cm}$ ~~~~t<T_{\max};$ \textbf{do}
\item {\textbf{4:} \hspace{0.5cm}\textbf{while}$~r\geqslant r_{\rm thre}$}\textbf{~do}
\item {\textbf{5:}\hspace{0.6cm}\text{Resampling based on}
\text{~$\boldsymbol{W}_{\mathrm{sub}}$ to obtain 
$~\boldsymbol{\varPhi }_{\mathrm{sub}}$};
\item $ \textbf{6:} \hspace{0.6cm}\text{Calculate} $
$~W_{\mathrm{sub},s}=\frac{P^{\left( \mathrm{GA} \right)}\left( \boldsymbol{X}_{\pi(s)} \right)}{\varXi \left( \boldsymbol{X}_{\pi(s)} \right)},~\forall s=1,\cdots,S$;
\item $\textbf{7:} \hspace{0.6cm}\text{Calculate}$
$r=\sum_{s=1}^S{\left| W_{\mathrm{sub},s}^{\left( z \right)}-W_{\mathrm{sub},s}^{\left( z-1 \right)} \right|}$
\item $\textbf{8:} \hspace{0.6cm}\text{Normalize the sample weights:}$
\item \hspace{1.5cm}$\hat{W}_{\mathrm{sub},s}=\frac{W_{\mathrm{sub},s}}{\sum_{s=1}^S{W_{\mathrm{sub},s}}},~\forall s=1,\cdots,S$;
\item $\textbf{9:} \hspace{0.4cm}\text{~~Update $ \boldsymbol{W }_{\mathrm{sub}}$; $~\boldsymbol{\varPhi }_{\mathrm{sub}}$; $~\varXi \left( \boldsymbol{X}_{\pi(s)} \right),\forall s=1\cdots,S$; }$
\item \textbf{10:} \hspace{0.4cm}$~z=z+1$;
\item $ \textbf{11:}\hspace{0.5cm}\textbf{end while}$;
\item {\textbf{12:}\text{~Employ the CVX toolbox to solve  Problem \textbf{P3};}
\item {\textbf{13:}\text{~$t=t+1$;}
\item $ \textbf{14: end while;}$
\item \textbf{15: Output:}$~|{\phi_k}|^2 =|{\phi_k^{\left( t\right)}}|^2,\forall k=1\cdots,K;$
}}}
\end{algorithmic}
\end{algorithm}

\textbf{Algorithm \ref{alg:low-complexity}} summarizes the proposed low-complexity algorithm developed in Section~\ref{sec:Algorithm_A}, which solves Problem \textbf{P3} iteratively from Step 3 to Step 14. Specifically, Steps 4 through 11 employ the adaptive importance sampling method to generate a sample set $\boldsymbol{\varPhi }_{\mathrm{sub}}$, and its corresponding weight set $\boldsymbol{W}_{\mathrm{sub}}$. In Step 12, we utilize the CVX toolbox to solve  Problem \textbf{P3} based on the sample set and weight set in each SCA iteration.  
With reference to Section \ref{Convergence Analysis: Algorithm 1}, the solution obtained by \textbf{Algorithm}~\ref{alg:low-complexity} converges to a solution that satisfies the KKT conditions of the original Problem~\textbf{P1}, and hence achieves local optimality for Problem~\textbf{P1}.

Like \textbf{Algorithm \ref{alg:alg1}}, the resolution of Problem \textbf{P3} through the SCA method and the CVX toolbox (e.g., the interior point method) dominates the complexity of \textbf{Algorithm \ref{alg:low-complexity}} per SCA iteration. Unlike \textbf{Algorithm \ref{alg:alg1}}, \textbf{Algorithm \ref{alg:low-complexity}} evaluates $\mu$ by only considering products of $S$ selected terms, i.e., 
$\prod\nolimits_{s=1}^S\big( \eta\big| \varPhi _{\mathrm{sub,}s} \big|^2+n_b \big)$ in \eqref{objective_u3} and \eqref{constraint_u3}.
 The complexity of \textbf{Algorithm~\ref{alg:low-complexity}} is $\mathcal{O} \left( S \cdot K^{4.5} \log \left( 1/\varepsilon _{\mathrm{SCA}} \right) \right)$. Considering $S\ll 2^K$, \textbf{Algorithm \ref{alg:low-complexity}} significantly alleviates the computational demands for power allocation in scenarios with a medium-to-large number of users, compared to \textbf{Algorithm~\ref{alg:alg1}}. 

\section{Simulations and Numerical Results}
\label{sec:Simulation}

Extensive simulations are conducted to evaluate the proposed coherent-state power allocation algorithms: \textbf{Algorithms~\ref{alg:alg1}} and \textbf{\ref{alg:low-complexity}}. The effects of atmospheric turbulence, background noise, and non-unit quantum efficiency are considered. 
The coherent states in atmospheric turbulence fluctuate, following the log-normal distribution due to scintillation phenomena \cite{quantum_log}. The PDF of the turbulence is 
$\mathrm{Pr}\big( h_{k,t} \big) =\frac{1}{h_{k,t}\sqrt{2{\pi \sigma _x}^2}}\exp \big[ -\frac{1}{2{\sigma _x}^2}\big( \ln h_{k,t}+\frac{{\sigma _x}^2}{2} \big) ^2 \big] $, where $\sigma _x$ is the intensity of turbulence \cite{quantum_log} and $h_{k,l}=\frac{1}{{d_k}^2}\left( \frac{\pi D_TD_R}{2\nu} \right) ^2$. $d_k$ is the distance between the BS and user $k$. $\nu$ is the wavelength. $D_T$ is the transmitter aperture diameter. $D_R$ is the receiver aperture diameter~\cite{Binary}. Unless otherwise stated, simulation parameters are provided in Table~\ref{tab:simulation}. 

Unlike classical AWGN channels, the conventional signal-to-noise ratio (SNR) is not a directly applicable metric in our NOMA-OQC system due to the fundamental differences, as summarized in Table I. The background noise \(n_b\) is a fixed parameter, and all benchmarks are evaluated under the identical channel conditions as our proposed \textbf{Algorithms 1} and~\textbf{2}, as listed in Table III. In our simulations, we vary the maximum receiving power \(|\phi_{\max,\text{BS}}|^2\) to evaluate the system performance across various communication scenarios. 

\begin{table}[t]
  \centering
  \renewcommand{\arraystretch}{1}
  \caption{Simulation Parameters of the System}
  \label{tab:simulation}
  \begin{tabular}{p{0.1\linewidth}p{0.48\linewidth}c}
    \hline
    \textbf{Notation} & \textbf{Description} & \textbf{Value}\\
    \hline
    $K$  & Total number of users & $16$\\
    $n_b$  & Background noise & $1.7$ \cite{backgroundnoise2}\\
    $\eta$  & Transmission coefficient of the lossy photon channel & $0.9$ \cite{lossy} \\
    $\nu$ & Transmission wavelength & $1550$ nm \cite{quantum_log}\\
    $D_T$& Transmitter aperture diameter & $10$ cm\\
    $D_R$& Receiver aperture diameter & $1$ m\\
    $\sigma _x$& Turbulence intensity & $0.3\sim0.5$ \\
    $d_k$& Distance between the BS and user $k$& $50\sim150$ m\\
    $\theta$ & Degrees-of-freedom of $t$-distribution & $1\sim 10$ \\
    $X_k$ & BPSK signal of the $k$-th user & $\left\{ -1,1 \right\}$ \\
    $\mathrm{Pr}\left( X_k
    \right)$ & The \textit{a priori} probability of user $k$ transmitting the signal & $\frac{1}{2}$ \\
    \hline
  \end{tabular}
\end{table}

To the best of our knowledge, no existing work is directly comparable with \textbf{Algorithms \ref{alg:alg1}} and~\textbf{\ref{alg:low-complexity}} due to our study of the new uplink NOMA-OQC  systems. We use the following baselines to serve as the benchmarks for the proposed algorithms: 
\begin{itemize}
\item {\em Orthogonal Multiple Access (OMA)}: This algorithm allocates orthogonal resources to different users to avoid multi-user interference. Each user transmits independently within their resources. The power allocation for the \(k\)-th user is based on the channel gain: \(|\phi_{k}|^2 = |\phi_{\max,\mathrm{BS}}|^2 \cdot \frac{h_k}{\sum_{i=1}^K h_i}\) \cite{OMA}. This approach offers a performance baseline for the proposed algorithms.

\item {\em Equal average number of photons (ENP)}: This algorithm assigns each user an equal average number of photons in their coherent states, i.e., $\left| \phi_1 \right|^2=\cdots = \left| \phi_K \right|^2$. It serves as a standard and implementation-friendly power allocation strategy \cite{ENP1}. 

\item {\em Interference assimilation (IA)}: 
When determining the data rate for the $k$-th user, this algorithm treats the signals from the other $(K-1)$ users as signal-independent
noise \cite{OMA1}. This method can be viewed as the benchmark for the proposed algorithms in solving Problem \textbf{P2} and \textbf{P3}. 
\end{itemize}

\begin{figure}
\centering
\includegraphics[width=3.5in]{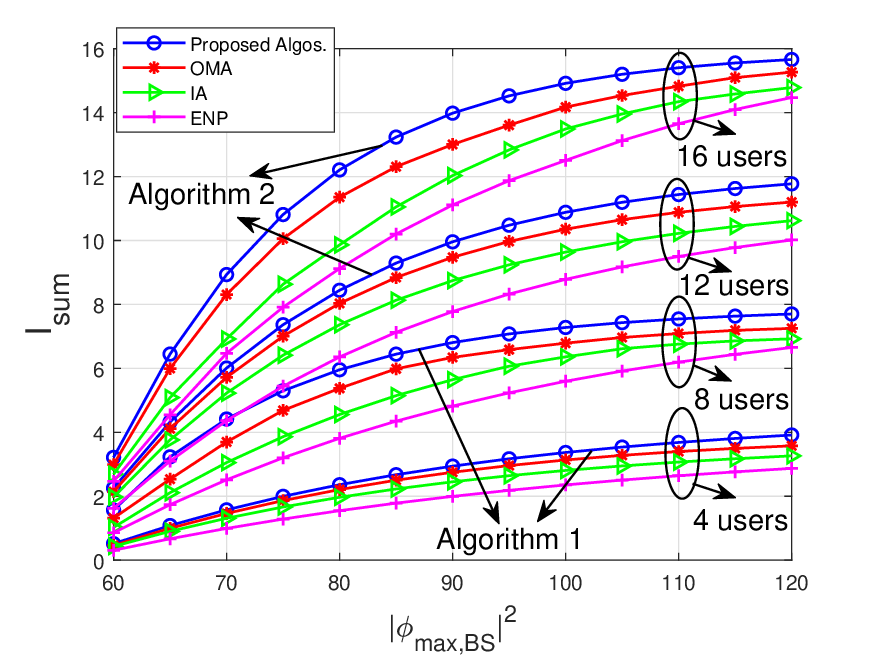}
\caption{Comparison of sum-rate between Algorithms 1/2, OMA, IA, and~ENP.}
\label{Comparison}
\end{figure}

Fig.~\ref{Comparison} evaluates \textbf{Algorithms~\ref{alg:alg1}} and \textbf{\ref{alg:low-complexity}} in comparison with the three benchmarks, where \textbf{Algorithm \ref{alg:alg1}} is employed when there are fewer users ($K\leqslant 8$) and \textbf{Algorithm \ref{alg:low-complexity}} is implemented when there are more users ($K>8$). \textbf{Algorithms~\ref{alg:alg1}} and \textbf{\ref{alg:low-complexity}} exhibit the fastest sum-rate growth among all schemes and achieve significant performance gains as $|\phi_{\max ,\mathrm{BS}}|^2$ increases. This indicates that, compared to OMA, IA, and ENP, \textbf{Algorithms~\ref{alg:alg1}} and \textbf{\ref{alg:low-complexity}} can effectively capture the impact of signal and noise through the asymptotically tight approximation of the sum-rate in \eqref{eq:Asymptotic_Sumrate}, which captures the signal characteristics after canceling the interference from the previous $(k-1)$ users when evaluating the data rate of the $k$-th user.

\begin{figure}[t]
\centering
\includegraphics[width=3.5in]{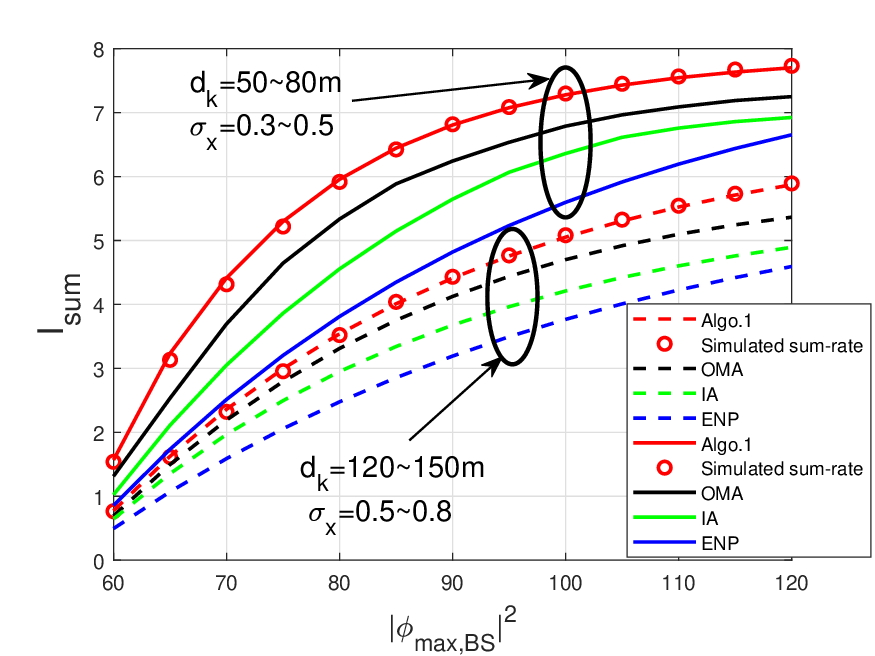}
\caption{Comparison of sum-rates for 8 users under different path loss and atmospheric turbulence scenarios.}
\label{fig:8_user_distance_turbulence}
\end{figure}

Fig. \ref{fig:8_user_distance_turbulence} evaluates \textbf{Algorithm 1}, compared to three benchmark schemes in an 8-user system with users situated at 50 to 150 m away from the BS. Each user experiences unique path loss characteristics based on distance, along with varying atmospheric turbulence conditions. We consider two distinct communication scenarios: Shorter-distance communications at 50 to 80 m under moderate turbulence conditions with the turbulence intensity $\sigma _x$ ranging from 0.3 to 0.5, and longer-distance communications at 120 to 150 m under stronger turbulence with $\sigma _x$ ranging from 0.5 to 0.8. 
     
As shown in  Fig. \ref{fig:8_user_distance_turbulence}, all four schemes demonstrate improved sum-rate performance as $|\phi_{\max ,\mathrm{BS}}|^2$ increases, with \textbf{Algorithm 1} exhibiting the highest sum-rate and the fastest growth rate. This superior performance demonstrates the enhanced capability of \textbf{Algorithm 1} to handle dynamic OQC factors.
The performance gains are attributed to the optimized power allocation that effectively compensates for channel variations and robust interference management, making it effective in scenarios with diverse user distributions and significant channel condition fluctuations. Also, the asymptotic approximation coincides with the simulation results, validating~\eqref{eq:Asymptotic_Sumrate}.

\begin{figure}
\centering
\includegraphics[width=3.5in]{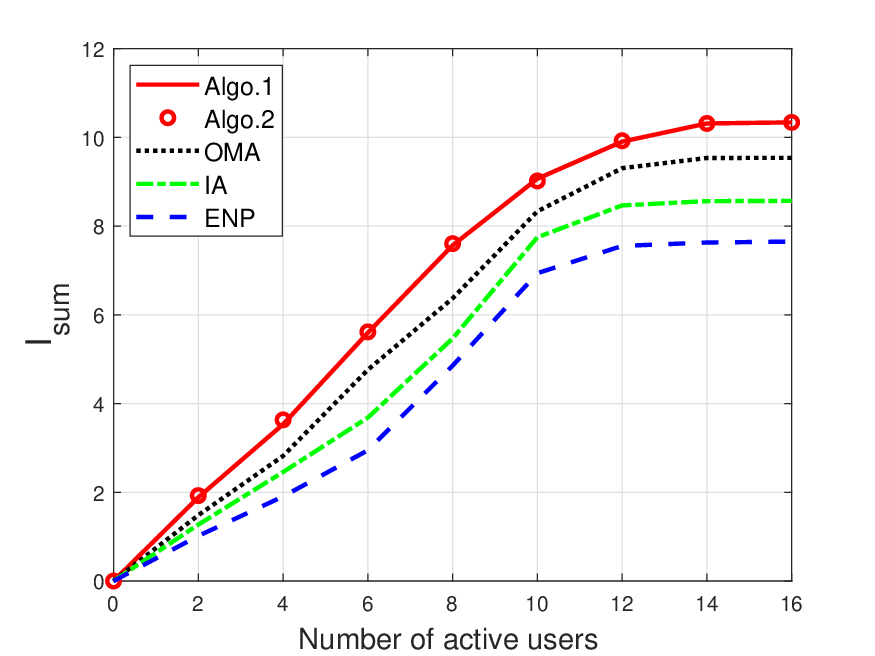}
\caption{Comparison of the sum-rate change with increasing number of users in the considered NOMA cluster.}
\label{number of active users.png}
\end{figure}

\begin{figure}
\centering
\includegraphics[width=3.5in]{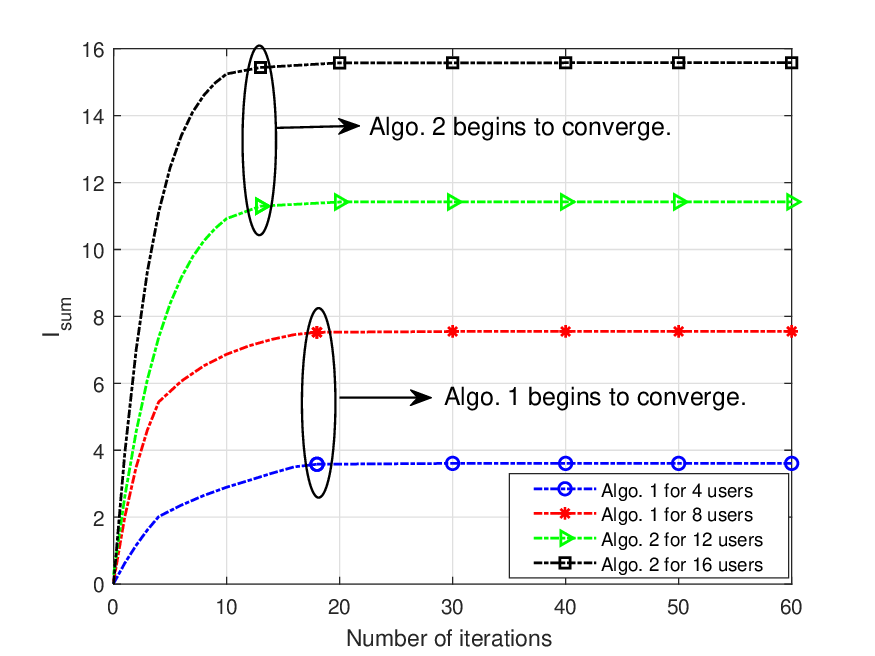}
\caption{ The convergence behaviors of \textbf{Algorithms 1} and \textbf{2}.}
\label{number of iteration}
\end{figure}

Fig.~\ref{number of active users.png} plots the variation of the sum-rate with an increasing number of users. It is evident that both \textbf{Algorithms~\ref{alg:alg1}} and~\textbf{\ref{alg:low-complexity}} substantially outperform the three benchmarks. Moreover, \textbf{Algorithms~\ref{alg:alg1}} and~\textbf{\ref{alg:low-complexity}} display excellent consistency in the sum-rate growth. \textbf{Algorithms~\ref{alg:low-complexity}} provides a polynomial-time alternative to the exponential-time \textbf{Algorithms~\ref{alg:alg1}}, reducing computational complexity substantially.
Fig.~\ref{number of iteration} demonstrates the convergence behaviors of \textbf{Algorithms \ref{alg:alg1}} and~\textbf{\ref{alg:low-complexity}} in systems with 4, 8, 12, or 16 users at $|\phi_{\max ,\mathrm{BS}}|^2=120$. Both algorithms satisfy KKT conditions, \textbf{Algorithm \ref{alg:low-complexity}} achieves faster convergence through its adaptive importance sampling method.

Fig.~\ref{fig:cap_reg} simulates the capacity region for a two-user uplink OQC system, where \( I_1 \) and \( I_2 \) denote the achievable rates of users 1 and 2, respectively. The region is bounded by \( I_1 < I_{1,\max} \), \( I_2 < I_{2,\max} \), and \( I_1 + I_2 < I_{\mathrm{sum}} \). The black line represents the capacity region of NOMA, while the green line represents the capacity region of OMA. Point A represents the case where user 2 achieves its maximum rate \( I_{2,\max} \), while user 1 achieves a rate of \( I_{1,\Delta}\). Point B represents the case where user 1 achieves its maximum rate \( I_{1,\max} \), while user 2 achieves a rate of \( I_{2,\Delta}\). The line segment AB represents the maximum sum-rate achievable by both users through different power allocations in the NOMA system. In the OMA system, the maximum sum rate is achieved at point C, where the rates of users 1 and 2 are lower than those achieved by NOMA at points A and B. Notably, the capacity region of OMA is a subset of that of NOMA, demonstrating the advantage of NOMA over OMA in both individual and sum-rates.

\begin{figure}[t]
\centering
\includegraphics[width=3.5in]{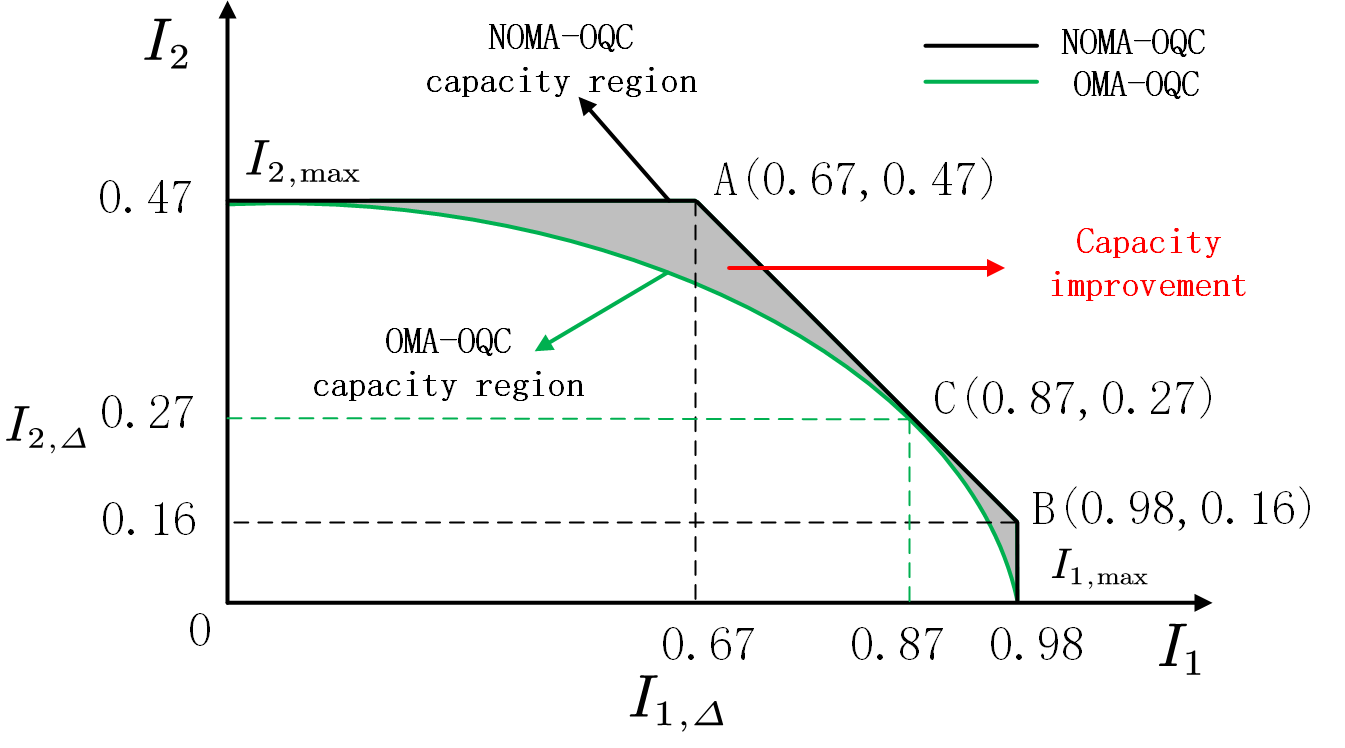}
    \caption{The two-user capacity regions of NOMA and OMA schemes in the proposed OQC communication system.}
\label{fig:cap_reg}
\end{figure}

\begin{figure}[htbp]
\centering
\includegraphics[width=3.5in]{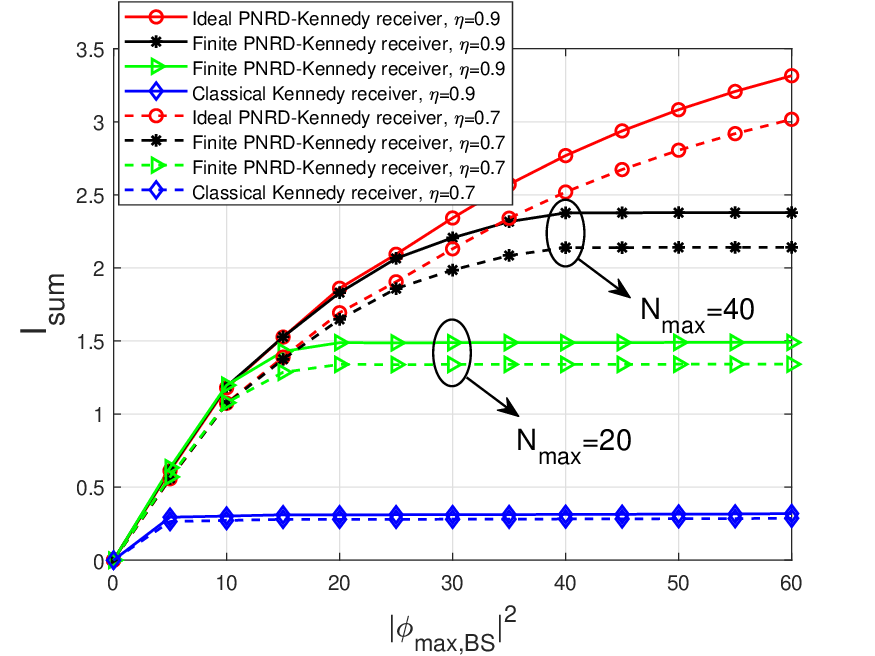}
\caption{Sum-rate comparison of 4-user uplink NOMA-OQC system for the ideal PNRD-Kennedy receiver, the finite PNRD-Kennedy receiver, and the classical Kennedy receiver.}
\label{fig:Different_receiver}
\end{figure}

\begin{figure}[t]
\centering
\includegraphics[width=3.5in]{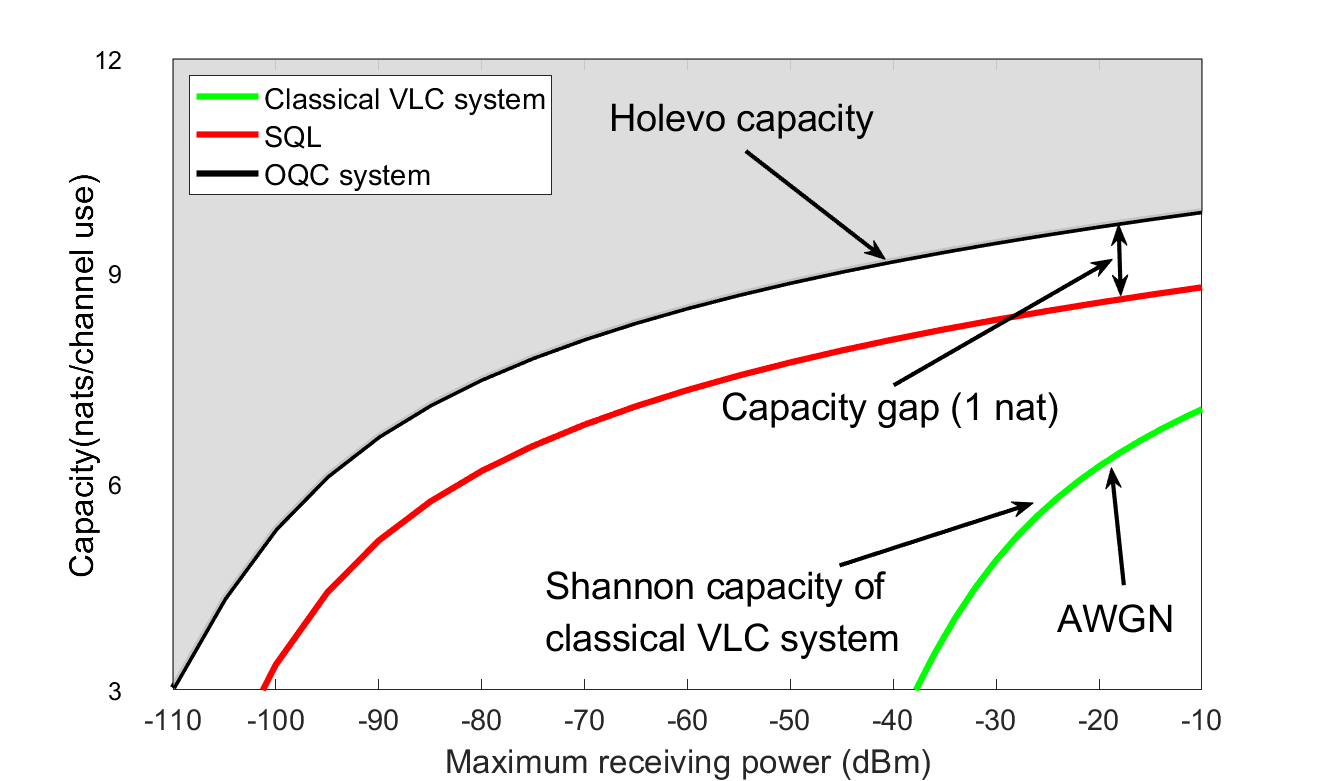}
\caption{Capacity comparison between OQC system, SQL and VLC system.}
\label{fig:Holevo_Shannon}
\end{figure}

Fig.~\ref{fig:Different_receiver} compares the sum-rate performance of three receiver architectures in a 4-user NOMA-OQC scenario: the ideal PNRD-Kennedy receiver, the finite PNRD-Kennedy receiver, and the classical Kennedy receiver. For the finite PNRD-Kennedy receiver, we consider the constraints of \(N_{\max}=25\) and \(N_{\max}=40\), which are representative of the maximum resolvable photon number \cite{Finite-PNRD1}. Furthermore, we evaluate the impact of the transmission coefficient \(\eta\) of the lossy photon channel by taking \(\eta = 0.7\) and \(\eta = 0.9\).
It is observed that the sum-rate of the PNRD-Kennedy receiver increases quickly with the growth of \(|\phi_{\max,\mathrm{BS}}|^2\), demonstrating a substantial advantage over the classical Kennedy receiver. The finite PNRD-Kennedy receiver has sum-rate converge under a lower received power than the ideal PNRD-Kennedy receiver.
Notably, the ideal PNRD-Kennedy receiver maintains superior performance robustness across different transmission coefficients of the lossy photon channel, outperforming traditional detection schemes under various lossy photon channel conditions.

We further compare the capacity regions of the proposed OQC system with both the SQL and a classical visible light communication (VLC) system. The capacity limit of OQC is governed by the Holevo bound $\chi = S\left(\sum_i p_i \rho_i\right) - \sum_i p_i S(\rho_i)$, where $S(\rho) = -\mathrm{tr}(\rho\log_2\rho)$ denotes the von Neumann entropy, with $\rho_i$ being the density operator of the $i$-th input quantum state transmitted with probability $p_i$ \cite{Quantum_Limits}. In contrast, the SQL refers to the performance limit achievable by directly measuring the modulated physical observables of coherent states through conventional homodyne detection, which is characterized by the standard Shannon capacity formula \cite[Eq.(32)]{Quantum_Limits}. The capacity of VLC is determined by the Shannon limit, which represents the classical capacity bound. 

    Fig.~\ref{fig:Holevo_Shannon} confirms the fundamental theoretical prediction \cite{Quantum_Limits} — the asymptotic gap between the Holevo capacity of OQC and the Shannon capacity approaches 1 nat in the high-power regime. For comparison with classical systems, we consider a VLC system operating within an 8 m × 8 m × 3 m room at a wavelength of 550 nm, with performance limited by AWGN of fixed noise power \(P_{\text{noise}} = 10^{-10} \, \text{W}\) (or $-70$ dBm) \cite{VLC1}. 
    Notably, the capacity of OQC system not only surpasses the SQL but also operates in a much lower power regime to achieve comparable performance compared to classical VLC system. 
    In other words, OQC provides higher channel capacity while outperforming the classical Shannon limit.
    Moreover, OQC inherently offers stronger security against eavesdropping, leveraging the no-cloning theorem — a fundamental advantage over classical VLC systems.

    \begin{figure}[t]
\centering
\includegraphics[width=3.5in]{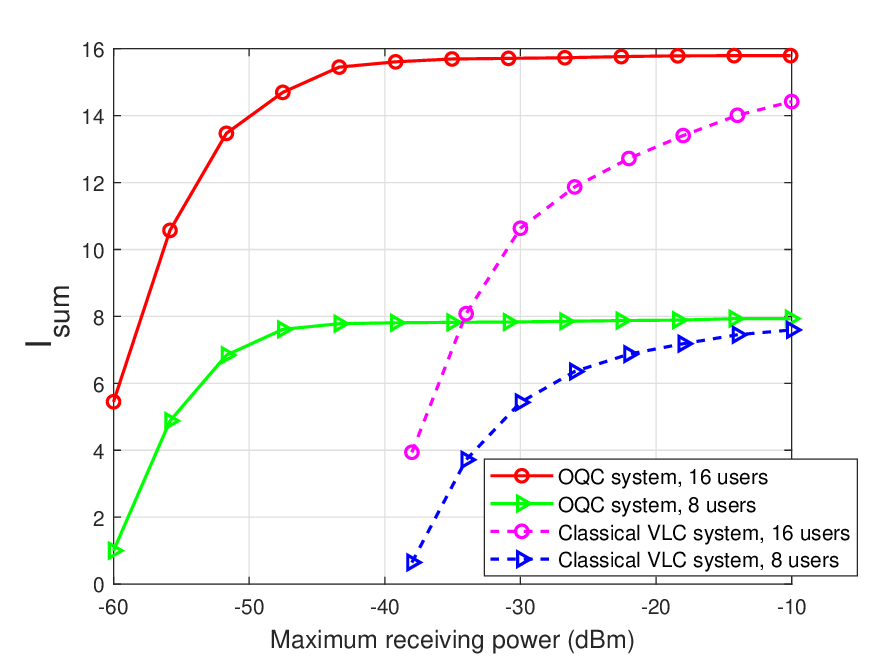}
\caption{Sum-rate comparison between the proposed quantum system and classical VLC system.}
\label{fig:OQC-VLC}
\end{figure}

\begin{figure}[t]
\centering
\includegraphics[width=3.5in]{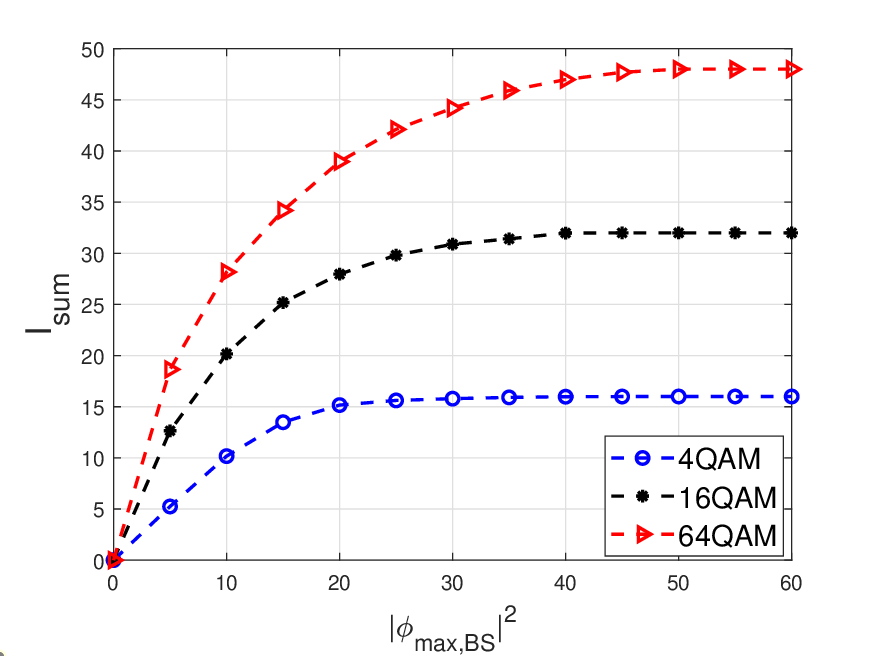}
\caption {Sum-rate comparison of 8-user uplink NOMA-OQC systems, where different QAM modulations are considered and Algorithm 1 is performed for coherent-state power allocation.}
\label{fig:QAM}
\end{figure}

     Fig.~\ref{fig:OQC-VLC} compares the sum-rate of our proposed quantum system with the classical NOMA-based VLC system under 8-user and 16-user communication scenarios. The considered VLC system operates within an indoor environment and employs intensity modulation and direct detection (IM/DD) limited by AWGN with noise power of \(P_{\text{noise}} = 10^{-10} \, \text{W}\) \cite{VLC}. The results demonstrate that our proposed scheme achieves a significantly higher sum-rate under 8-user and 16-user communication scenarios. This performance gain highlights the superior multi-user scalability and spectral efficiency of the NOMA-OQC framework in practical communication systems. 
    
    Last but not least, the proposed algorithm can be extended to higher-order modulations, e.g., quadrature amplitude modulation (QAM), with the sum-rate analyzed in the same way as with BPSK; see Appendix D. Fig.~\ref{fig:QAM} plots numerically the sum-rate for QAM under an 8-user scenario, by varying the maximum effective received power at the receiver. As expected, the sum-rates increase significantly with the modulation order. The sum-rates converge when $|\phi_{\max, \mathrm{BS}}|^2>50$.

\section{Conclusion}
\label{sec:Conclusion}
We have analyzed the sum-rate of uplink NOMA-OQC systems using coherent states over a lossy photon channel, considering atmospheric turbulence and non-negligible background noises. We have also presented a new approach to allocating the powers of the coherent states within the system. Specifically, we have derived the asymptotic sum-rate of the system, with which we optimized the average number of photons (or powers) of the coherent states emitted by the users. Using variable substitution and SCA, a new power allocation algorithm has been developed to iteratively convexify and maximize the asymptotic sum-rate. 
We have further reduced the complexity of the algorithm using adaptive importance sampling, to cater to a medium-to-large number of users. Simulations have demonstrated that the proposed algorithms significantly enhance the sum-rate of the uplink NOMA-OQC systems by over 20\%, compared to their potential alternatives. 
Future work will explore advanced receiver designs, such as the Dolinar receiver with feedforward detection, to further approach the Helstrom bound for multi-user coherent-state communications.

\appendices
\section{Proof of Theorem 1}

By the definition of Shannon entropy, we derive the 
entropy $H(Y)$ and the conditional entropy $H\left(Y|X_1,\cdots,X_K\right)$, as given in \eqref{eq:SumRate,2}. Specifically,
\begin{equation}
\label{entropy}
\begin{aligned}
&H\big(Y\big) =-\sum\limits_{y=0}^{+\infty}{\Big\{{\sum_{i=1}^{2^K}\Pr\big(\boldsymbol{X}_i\big)}\Pr\big(Y|\boldsymbol{X}_i\big)\big.} \\
&\quad\quad\quad\quad\quad\quad\quad\big. \times \log_2\big[\sum_{i=1}^{2^K}\Pr\big(\boldsymbol{X}_i\big)\Pr\big(Y|\boldsymbol{X}_i\big)\big]\Big\},
\end{aligned}
\end{equation}
\begin{equation}
\label{conditional_entropy}
\begin{aligned}
&H\left(Y|X_1,\cdots,X_K\right)=-\sum\limits_{y=0}^{+\infty}\sum_{i=1}^{2^K}\Pr\big(\boldsymbol{X}_i\big)\Pr\big(Y|\boldsymbol{X}_i\big)
\\
&
\quad\quad\quad\quad\quad\quad\quad\quad\quad\quad\quad \times \log_2\big[\Pr\big(Y|\boldsymbol{X}_i\big)\big],
\end{aligned}
\end{equation}
where $\boldsymbol{X}_i, \forall i=1,\cdots, 2^K$ enumerate all $2^K$ possibilities of $\boldsymbol{X}=\left[ X_1,\cdots,X_K \right] ^T$ with $X_k\in\{-1,1\}$ being the BPSK signal of user $k, \forall k=1,\cdots,K$.

By substituting \eqref{entropy}, \eqref{conditional_entropy} and \eqref{eq:basePr_group} into \eqref{eq:SumRate,2}, we have the sum-rate in \eqref{eq:SumRate_appendix}. In particular, leveraging the definitions of entropy and conditional entropy yields \eqref{eq:SumRate_appendix_a}. 
By substituting \eqref{eq:basePr_group} into \eqref{eq:SumRate_appendix_a} and then simplifying, we arrive at \eqref{eq:SumRate_appendix_d}, where $\varPhi_i, \forall i=1,\cdots, 2^K$ enumerate all $2^K$ possibilities of $|\varPhi \rangle =|\sum_{k=1}^K{\phi _k}\rangle$ with $\phi _k$ being the received coherent state of user $k$ corresponding to BPSK signal $X_k$, $\forall k=1,\cdots,K$.
\begin{figure*}[ht]
\begin{subequations}
\label{eq:SumRate_appendix}
\begin{flalign}
I_{\mathrm{sum}}&=-\sum\limits_{y=0}^{+\infty}{\bigg\{{\sum_{i=1}^{2^K}\Pr\big(\boldsymbol{X}_i\big)}\Pr\big(Y|\boldsymbol{X}_i\big)\log_2\Big[\sum_{i=1}^{2^K}\Pr\big(\boldsymbol{X}_i\big)\Pr\big(Y|\boldsymbol{X}_i\big)\Big]\bigg\}}+\sum\limits_{y=0}^{+\infty}\sum_{i=1}^{2^K}{\bigg\{\Pr\big(\boldsymbol{X}_i\big)\Pr\big(Y|\boldsymbol{X}_i\big)\log_2\big[\Pr\big(Y|\boldsymbol{X}_i\big)\big]\bigg\}}\label{eq:SumRate_appendix_a}\\ 
&=K + \frac{1}{2^K\ln 2}\cdot\sum_{y=0}^{+\infty} \bigg\{
\sum_{i=1}^{2^K} \Big[ \frac{\big( \eta\big| \varPhi_i \big|^2+n_b \big)^y}{y!} \exp \big( -\eta\big| \varPhi_i \big|^2-n_b \big) \cdot\ln \frac{\big( \eta\big| \varPhi_i \big|^2+n_b \big)^y}{y!} \exp \big( -\eta\big| \varPhi_i \big|^2-n_b \big) \Big] \big.\notag\\
&\big.\quad\quad\quad\quad\quad\quad- \sum_{i=1}^{2^K} \frac{\big( \eta\big| \varPhi_i \big|^2+n_b \big)^y}{y!} \exp \big( -\eta\big| \varPhi_i \big|^2-n_b \big)\cdot\ln \Big[ \sum_{i=1}^{2^K} \frac{\big( \eta\big| \varPhi_i \big|^2+n_b \big)^y}{y!} \exp \big( -\eta\big| \varPhi_i \big|^2-n_b \big) \Big]
\bigg\}.
\label{eq:SumRate_appendix_d}
\end{flalign}
\end{subequations}
\hrulefill
\end{figure*}

\section{Proof of Lemma 1}

According to \cite{Yongkang2} and \cite{Yongkang}, the probability mass function of the Poisson distribution, i.e., \eqref{eq:basePr_group}, can be approximated by a PDF of the Gaussian distribution as $n_b \rightarrow +\infty$. Following  \cite[eq. (21)]{Yongkang}, we have
\begin{equation}
\label{eq:GA}
\begin{aligned}
P( \boldsymbol{X} ) \xrightarrow{n_b\rightarrow \infty}P^{( \mathrm{GA} )}( \boldsymbol{X} ) =\frac{\exp \Big[ -\frac{\big( y-\displaystyle \eta\big| \varPhi \big|^2-n_b \big) ^2}{2\big( \displaystyle \eta\big| \varPhi \big|^2+n_b \big)} \Big]}{\sqrt{2\pi \big( \displaystyle \eta\big| \varPhi \big|^2+n_b \big)}},
\end{aligned}
\end{equation}
where the superscript ``$^{( \mathrm{GA})}$'' represents “Gaussian approximation”. As a result, we can have
 \begin{equation}
\label{deqn_ex1}
\begin{aligned}
&I_\mathrm{sum}\xrightarrow{n_b\rightarrow +\infty}\tilde{I}_\mathrm{sum}^{( \mathrm{GA} )}\\
&=K+\frac{1}{2^K\ln 2} \cdot\bigg\{ \int\limits_0^{\infty}{\Big[ \sum_{i=1}^{2^K}{P^{( \mathrm{GA} )}\big( \boldsymbol{X}_i \big) \ln P^{( \mathrm{GA} )}\big( \boldsymbol{X}_i \big)} \Big] dy}\big.\\	
&\big.\quad-\int\limits_0^{\infty}{ \sum_{i=1}^{2^K}{P^{( \mathrm{GA} )}\big( \boldsymbol{X}_i \big)} \cdot \ln \Big[ \sum_{i=1}^{2^K}{P^{( \mathrm{GA} )}\big( \boldsymbol{X}_i \big)} \Big] dy}\bigg\}.
\end{aligned}
\end{equation}
Next, we integrate the RHS of \eqref{deqn_ex1}. For illustration convenience, we define part of \eqref{deqn_ex1} as
\begin{align}
\label{eq:varphi_function}
\varphi \big( \boldsymbol{\varPhi } \big) \stackrel{\triangle}{=}&\int\limits_0^{\infty}{\Big[ \sum_{i=1}^{2^K}{P^{( \mathrm{GA})}\big( \boldsymbol{X}_i \big) \ln P^{( \mathrm{GA} )}\big( \boldsymbol{X}_i \big) \Big]} dy};\\
\label{eq:f_omega}
f\big( \boldsymbol{\omega } \big) \stackrel{\triangle}{=}&\int\limits_0^{\infty}{ \sum_{i=1}^{2^K}{P^{( \mathrm{GA} )}\big( \boldsymbol{X}_i \big)}\cdot \ln \Big[ \sum_{i=1}^{2^K}{P^{( \mathrm{GA} )}\big( \boldsymbol{X}_i \big)} \Big] dy},
\end{align}
where $\boldsymbol{\omega }=\big[ \omega_1,\cdots,\omega_{2^K} \big) ]$, as elaborated in detail in the subsequent paragraphs.

 We rewrite \eqref{eq:varphi_function} as 
\begin{align}
&\varphi \big( \boldsymbol{\varPhi } \big) =\int\limits_0^{\infty}\bigg\{ 
\frac{\exp \Big[ -\frac{\big( y-n_b \big) ^2}{2n_b} \Big]}{\sqrt{2\pi n_b}}\ln \frac{\exp \Big[ -\frac{\big( y-n_b \big) ^2}{2n_b} \Big]}{\sqrt{2\pi n_b}} +\cdots + \big.\nonumber \\	
&\big. \frac{\exp \Big[ -\frac{\big( y-\eta\big| \varPhi_{2^K} \big|^2-n_b \big) ^2}{2\big( \eta\big| \varPhi_{2^K} \big|^2+n_b \big)} \Big]}{\sqrt{2\pi \big( \eta\big| \varPhi_{2^K} \big|^2+n_b \big)}}\ln \frac{\exp \Big[ -\frac{\big( y-\eta\big| \varPhi_{2^K} \big|^2-n_b \big) ^2}{2\big( \eta\big| \varPhi_{2^K} \big|^2+n_b \big)} \Big]}{\sqrt{2\pi \big( \eta\big| \varPhi_{2^K} \big|^2+n_b \big)}} 
\bigg\}dy,\label{deqn_ex1_0}
\end{align}
where there are $2^K$ terms, with the first term rewritten as
\begin{subequations}
\label{deqn_ex1_1}
\begin{flalign}
&\int_0^{\infty}{\frac{\exp \Big[ -\frac{\big( y-n_b \big) ^2}{2n_b} \Big]}{\sqrt{2\pi n _b}}\ln \bigg\{ \frac{\exp \Big[ -\frac{\big( y-n_b \big) ^2}{2n_b} \Big]}{\sqrt{2\pi n _b}} \bigg\}}dy\\
&=-\frac{\ln \big( 2\pi n_b \big)}{2}\int_0^{\infty}{\frac{\exp \big[ -\frac{\big( y-n_b \big) ^2}{2n_b} \big]}{\sqrt{2\pi n_b}}}dy\notag\\
&\quad-\int_0^{\infty}{\frac{\exp \Big[ -\frac{\big( y-n_b \big) ^2}{2n_b} \Big]}{\sqrt{2\pi n_b}}\frac{\big( y-n_b \big) ^2}{2n_b}dy}\\
&=-\frac{\ln \big( 2\pi n_b \big) +1}{2}\Big[ 1-Q\big( \sqrt{n_b} \big) \Big] +\frac{1}{2\sqrt{2\pi}}\sqrt{n_b}\exp \big( -\frac{n_b}{2} \big)\notag\\
&\xrightarrow{n_b\rightarrow +\infty}-\frac{\ln \big( 2\pi n_b \big) +1}{2},
\end{flalign}
\end{subequations}
where \eqref{deqn_ex1_1} follows from the partial integration method. $Q\left( \cdot \right) $ is the Q-function.
By the L’Hospital’s rule, as $n_b\rightarrow +\infty $, $\frac{\ln \left( 2\pi n _b \right) +1}{2}Q\left( \sqrt{n_b} \right) \rightarrow 0$ and $\sqrt{n _b}\exp \left( -\frac{n_b}{2} \right) \rightarrow 0$.

Similarly, we can evaluate the remaining $\big( 2^K-1 \big) $ integral terms in \eqref{deqn_ex1_0}. As a result, \eqref{deqn_ex1_0} can be rewritten as
\begin{equation}
\label{deqn_ex1_}
\varphi \big( \boldsymbol{\varPhi }\big)\! =\!
- \frac{1}{2}  \ln\! \Big[ \big( 2\pi \big) ^{2^K}n_b\prod\nolimits_{i=1}^{2^K}{\big( \eta\big| \varPhi_i \big|^2\!\!+\!n_b \big)} \Big] \!\!-\!2^{K\!-\!1}.
\end{equation}
It is challenging to calculate \eqref{eq:f_omega} due to the involvement of $2^K$ Gaussian PDFs and logarithm operation. Hence, we opt to determine its lower bound, as follows:
    \begin{align}
&f\big( \boldsymbol{\omega }_{\mathrm{low}} \big) 
=\int\limits_0^{\infty}{ \sum_{i=1}^{2^K}{P_{\mathrm{low}}^{( \mathrm{GA} )}\big( \boldsymbol{X}_i \big)} \cdot\ln \Big[ \sum_{i=1}^{2^K}{P_{\mathrm{low}}^{( \mathrm{GA} )}\left( \boldsymbol{X}_i \right)} \Big] dy}\nonumber\\
&=\!\!\int_0^{\infty}\!\!{\bigg\{	
\frac{\exp \Big[ \!-\!\frac{\big( y\!-\!n_b \big) ^2}{2n_b} \Big]}{\sqrt{2\pi n_b}}\!+\!\cdots	\!+\!\frac{\exp \Big[ \!-\!\frac{\big( y\!-\!\eta\big |\varPhi_{2^K}\big|^2\!-\!n_b \big) ^2}{2\big( \eta\big |\varPhi_{2^K}\big|^2\!+\!n_b \big)} \Big]}{\sqrt{2\pi n_b}} \bigg\}}\nonumber\\
&\quad \ln\bigg\{	\frac{\exp \left[ \!-\!\frac{\left( y\!-\!n_b \right) ^2}{2n_b} \right]}{\sqrt{2\pi n_b}}\!+\!\cdots	\!+\!\frac{\exp \left[ \!-\!\frac{\left( y\!-\!\eta\left |\varPhi_{2^K}\right|^2\!-\!n_b \right) ^2}{2\left( \eta\left |\varPhi_{2^K}\right|^2\!+\!n_b \right)} \right]}{\sqrt{2\pi n_b}}\bigg\} dy.
\label{eq_lowbound}
\end{align}
After performing variable substitution 
and mathematical manipulation on \eqref{eq_lowbound}, the lower bound $f\left( \boldsymbol{\omega }_{\mathrm{low}} \right)$ is rewritten as 
\begin{subequations}
    \label{first}
    \begin{flalign}
f&\big( \boldsymbol{\omega }_{\mathrm{low}} \big) 
=\int\limits_0^{\infty}{ \sum_{i=1}^{2^K}{P_{\mathrm{low}}^{( \mathrm{GA} )}\big( \boldsymbol{X}_i \big)} \ln \Big[ \sum_{i=1}^{2^K}{P_{\mathrm{low}}^{( \mathrm{GA} )}\big( \boldsymbol{X}_i \big)} \Big] dy}\\
&\xrightarrow{n _b\rightarrow +\infty}
=\int_{-\infty}^{\infty}{ \sum_{i=1}^{2^K}\exp \Big[ -\big( u_{\mathrm{low}}+\omega _{\mathrm{low},i} \big) ^2 \Big] }\notag\\
&\quad\quad\quad\quad\cdot\ln \bigg\{ \sum_{i=1}^{2^K}\exp \Big[ -\big( u_{\mathrm{low}}+\omega _{\mathrm{low},i} \big) ^2 \Big]  \bigg\} du_{\mathrm{low}}\notag\\
&\quad\quad\quad-\frac{2^{K-1}\ln \big( 2\pi n_b \big)}{\sqrt{2n_b}}, \label{first_c}
\end{flalign}
\end{subequations}
where $u_{\mathrm{low}}=\frac{1}{\sqrt{2n_b}}\big( y-\frac{1}{2}\displaystyle\sum_{i=1}^{2^K}{\eta\big| \varPhi_i \big|^2}-n_b \big)$.

Let $g(\boldsymbol{\omega})$ denote the integral term in \eqref{first_c}, we have 
\begin{subequations}
\label{function1}
\begin{flalign}
g\big( \boldsymbol{\omega } \big) 
&=\int_{-\infty}^{\infty}{ \sum_{i=1}^{2^K}\exp \Big[ -\big( u+\omega _i \big) ^2 \Big]  }\notag\\
&\quad\quad\quad\quad\cdot\ln \bigg\{ \sum_{i=1}^{2^K}\exp \Big[ -\big( u+\omega _i \big) ^2 \Big]  \bigg\} du\label{function2_b}\\
&\geqslant \ln \int_{-\infty}^{\infty}{ \sum_{i=1}^{2^K}\exp \big[ -\big( u+\omega _i \big) ^2 \big]  }\notag\\
&\quad\quad\quad\quad\quad\cdot\bigg\{ \sum_{i=1}^{2^K}\exp \big[ -\big( u+\omega _i \big) ^2 \big] \bigg\} du,\label{function2_c}
\end{flalign}
\end{subequations}
where \eqref{function2_c} is obtained by applying Jensen’s inequality to \eqref{function2_b},
since \eqref{function2_b} is convex.

In \eqref{function2_c}, the product of the two sums of $2^K$ exponential functions results in a total of $2^{2K}$ exponential functions inside the integral. The $k$-th term in \eqref{function2_c} can be expressed as
\begin{equation}
\label{function_g_w}
\begin{aligned}
&\int_{-\infty}^{\infty}{\exp \Big[ -\big( u+\omega _k \big) ^2 \Big] \exp \Big[ -\big( u+\omega _{k+1} \big) ^2 \Big] \,\,}du\\
&=\sqrt{\frac{\pi}{2}}\cdot \exp \Big[ -\big( \frac{\omega _k-\omega _{k+1}}{2} \big) ^2 \Big]. 
\end{aligned}
\end{equation}
Likewise, we derive the remaining $\big (2^{2K}-1\big)$ terms of~\eqref{function2_c}. 

By substituting \eqref{function_g_w} and the corresponding expressions for the remaining $\big (2^{2K}-1\big)$ terms into \eqref{function2_c}, we can obtain
\begin{equation}
\label{function2_d}
\begin{aligned}
&\tilde{g}\left( \boldsymbol{\omega } \right)\! =\!
\ln\! \bigg\{ \!\sqrt{\frac{\pi}{2}}\!\bigg[ {\sum_{i=1}^{2^K}{\sum_{j =1}^{2^K}\!{{\binom{2^K}{i}}\!{ \binom{2^K}{j}}\!\exp \!\Big( \!\frac{{\omega}_{i}\!-\!{\omega }_{j}}{2} \!\Big) ^2}}} \bigg] \!\bigg\}\!+\!\xi,
\end{aligned}
\end{equation}
where $\xi$ is a constant shift applied to the lower bound \eqref{function2_c} to make it an accurate approximation to $g\big( \boldsymbol{\omega } \big)$ in \eqref{function2_b}~\cite{approximate}.

After combining identical exponential terms, there are $\big( 2^{K}+{ \binom{2^K}{j}}\big)$ different exponential terms inside $\ln \left( \cdot \right) $.
Substituting \eqref{function2_d} into \eqref{first}, we obtain the lower bound of $f\left( \boldsymbol{\omega } \right) $, as given in \eqref{eq:SumRate_low}.
Substituting \eqref{deqn_ex1_} and \eqref{eq:SumRate_low} into \eqref{deqn_ex1}, we can obtain the lower bound of the sum-rate, as given in \eqref{Sum_rate_lowbound}. As $n_b\rightarrow +\infty $, we have $\omega _{\mathrm{low,}k}\rightarrow 0$. As a result, the lower bound $I_{\mathrm{low}}^{\left( \mathrm{GA} \right)}$  converges asymptotically to $K+\frac{\varphi \left( \boldsymbol{\varPhi } \right) -f\left( 0 \right)}{2^K\ln 2}$.

\section{Proof of Lemma 2}
By performing variable substitution and mathematical manipulation on \eqref{eq:f_omega}, we have 
\begin{align}\label{second_c}
&f\big( \boldsymbol{\omega }_{\mathrm{up}} \big) 
=\int\limits_0^{\infty}{ \sum_{i=1}^{2^K}{P_{\mathrm{up}}^{( \mathrm{GA} )}\big( \boldsymbol{X}_i \big)}  \ln \Big[ \sum_{i=1}^{2^K}{P_{\mathrm{up}}^{( \mathrm{GA} )}\big( \boldsymbol{X}_i \big)} \Big] dy}
\nonumber
\\
&\xrightarrow{n _b\rightarrow +\infty}
\int_{-\infty}^{\infty}{ \sum_{i=1}^{2^K}\exp \Big[ -\big( u_{\mathrm{up}}+\omega _{\mathrm{up},i} \big) ^2 \Big] }\notag\\
&\quad\quad\quad\quad\cdot\ln \bigg\{ \sum_{i=1}^{2^K}\exp \Big[ -\big( u_{\mathrm{up}}+\omega _{\mathrm{up},i} \big) ^2 \Big] \bigg\} du_{\mathrm{up}}
\nonumber
\\
&\quad\quad-\frac{2^{K-1}\ln \Big[ 2\pi \big(\eta \big| \displaystyle\sum_{k=1}^{K}{\phi_k} \big|^2+n_b \big) \Big]}{\sqrt{2\big(\eta \big| \displaystyle\sum_{k=1}^{K}{\phi_k} \big|^2+n_b \big)}},
\end{align}
where
$u_{\mathrm{up}}=(y-\frac{1}{2}\displaystyle\sum_{i=1}^{2^K}{\eta\left| \varPhi_i \right|^2}-n_b)
\Bigg/
\sqrt{2(\eta | \displaystyle\sum_{k=1}^{K}{\phi_k} |^2+n_b )}$.
Substituting \eqref{function2_d} into \eqref{second_c}, we derive the upper bound of $f\left( \boldsymbol{\omega } \right) $, as given in \eqref{eq:SumRate_up}.
By substituting \eqref{deqn_ex1_} and \eqref{eq:SumRate_up} into \eqref{deqn_ex1}, we obtain the upper bound of the sum-rate, as given in \eqref{Sum_rate_upbound}. As $n_b\rightarrow +\infty $, we have $\omega _{\mathrm{up,}k}\rightarrow 0$. As a result, the upper bound $I_{\mathrm{up}}^{\left( \mathrm{GA} \right)}$ converges asymptotically to $K+\frac{\varphi \left( \boldsymbol{\varPhi } \right) -f\left( 0 \right)}{2^K\ln 2}$.

\section{Extension to Quadrature Amplitude Modulation}

We investigate the extension from BPSK to QAM in uplink OQC systems, derive the sum-rate expression for multi-user QAM with coherent states, and analyze its convergence properties when employing the square root measurement (SRM) technique for detecting the higher-order modulated coherent states \cite{QAM}. The classical information from the users is mapped into a vector $\boldsymbol{X}^{\text{Q}}=\left[X_1^{\text{Q}},\cdots,X_K^{\text{Q}}\right]^T$, where $X_k^{\text{Q}} \in \mathbb{R}$ represents the QAM signal of user $k$. 
These signals are modulated using QAM and transmitted via coherent states.
Assuming equal probability for each symbol, the QAM constellation consists of $M=L^2$ equidistant points arranged in a square grid on the complex plane. Let $\Omega $ denote the set collecting the indices of the QAM modulation symbols:
 \begin{equation}
 \label{eq:index}
    \Omega \!=\!\left\{ -\left( L\!-\!1 \right) \!+\!2\left( j\!-\!1 \right) |j\!=\!1,\!\cdots \!,L \right\} , L\!=\!2,3,\!\cdots\!.
    \end{equation}
According to the principles of quantum mechanics, in a multi-user quantum coherent state communication system, the coherent state modulated by QAM can be expressed as:
 \begin{equation}
 \label{eq:coherent_state}
    |\alpha _{p,q}\rangle =|\alpha \left( p+qi \right) \rangle , p,q\in \Omega.
    \end{equation}

At the receiver, we employ the SRM method to obtain the detection probability of the received quantum states. SRM is an optimal measurement strategy that minimizes the error probability when distinguishing between a set of quantum states, making it advantageous for applications in QAM modulation. 
It involves constructing a set of POVM elements
derived from the square root of the Gram matrix associated with the received quantum states. The Gram matrix captures the inner products between the quantum states, and its square root defines the measurement operators. 

We first construct the Gram matrix $G$ with elements:
\begin{equation}
\begin{aligned}
 \label{eq:matrix_elements}
    &\langle \alpha _{pq}|\alpha _{p^{\prime}q^{\prime}}\rangle
    =\langle \alpha \big( p+qi \big) |\alpha \big( p^{\prime}+q^{\prime}i \big) \rangle \\
    &=\!\exp\! \Big\{ \!-\!\frac{1}{2}\alpha ^2\big[ \big( p^{\prime}\!-\!p \big) ^2\!+\!\big( q^{\prime}\!-\!q \big) ^2\!-\!2i\big( p^{\prime}p\!-\!q^{\prime}q \big) \big] \Big\}, 
    \end{aligned}
    \end{equation}
where $p^{\prime},q^{\prime}\in \Omega $ are the coordinates on the complex plane similar to $p$ and $q$. We use Lexicographic order to construct the Gram matrix based on these four elements, $p,q,p^{\prime}$, and $q^{\prime}$. 

Next, we perform eigenvalue decomposition on the Gram matrix $G$, as given by
\begin{equation}
 \label{eq:eigenvalue}
    G=V\varLambda _GV^*=\sum_{i=1}^L{{\sigma _i}^2|v_i\rangle \langle v_i|},
    \end{equation}
where ${\sigma _i}^2$ is the eigenvalue, and $|v_i\rangle $ is the eigenvector. 

Using the eigenvalue decomposition, we construct the SRM measurement operator, as given by
\begin{equation}
   M_y = G^{-1/2} |\alpha_y\rangle \langle \alpha_y| G^{-1/2},
    \end{equation}
where $G^{-1/2}$ is the inverse square root of the Gram matrix, and $M_y$ satisfy the completeness constraint, i.e.,\(\sum_y M_y = I\).

Finally, for the transmitted signal $X_i^{\text{Q}}, i=1,2,\cdots ,K$, the conditional probability of detecting the result $y$ is given by 
\begin{subequations}
\label{eq:conditional_probability}
\begin{flalign}
    P\big( \boldsymbol{X}^{\text{Q}} \big) 
    &\triangleq \mathrm{Pr}(Y=y|X_i^{\text{Q}})\\
    &= \langle \alpha_i | M_y | \alpha_i \rangle \\
    &= \big| \langle \alpha_y | G^{-1/2} | \alpha_i \rangle \big|^2.
    \end{flalign}
\end{subequations}

We further derive the entropy and conditional entropy of the received signals for the QAM modulated coherent states:
\begin{equation}
\label{entropy_QAM}
\begin{aligned}
&H\big(Y\big) =-\sum\limits_{y=0}^{+\infty}{\bigg\{{\sum_{i=1}^{m^K}\Pr\big(\boldsymbol{X}_i^{\text{Q}}\big )}\Pr\left(Y|\boldsymbol{X}_i^{\text{Q}}\right)\big.} \\
&\quad\quad\quad\quad\big. \times \log_2\big[\sum_{i=1}^{m^K}\Pr\big(\boldsymbol{X}_i^{\text{Q}}\big)\Pr\big(Y|\boldsymbol{X}_i^{\text{Q}}\big)\Big ]\bigg\},
\end{aligned}
\end{equation}
\begin{equation}
\label{conditional_entropy_QAM}
\begin{aligned}
&H\big(Y|X_1^{\text{Q}},\cdots,X_K^{\text{Q}}\big)=-\sum\limits_{y=0}^{+\infty}\sum_{i=1}^{m^K}\Pr\big(\boldsymbol{X}_i^{\text{Q}}\big)
\\
&
\quad\quad\quad\quad\times\Pr\big(Y|\boldsymbol{X}_i^{\text{Q}}\big) \log_2\Big[\Pr\big(Y|\boldsymbol{X}_i^{\text{Q}}\big)\Big],
\end{aligned}
\end{equation}
where $m$ is the modulation order, and $\boldsymbol{X}_i^{\text{Q}}, \forall i=1,\cdots, m^K$ enumerate all $m^K$ possibilities of $\boldsymbol{X}^{\text{Q}}=\left[ X_1^{\text{Q}},\cdots,X_K^{\text{Q}} \right] ^T$ with $X_k$ being the QAM signal of user $k, \forall k=1,\cdots,K$.

The mutual information between the input and output quantum states characterizes the transmission rate, 
and 
establishes that the uplink multi-user sum-rate \eqref{eq:SumRate_H} retains validity for QAM-modulated coherent states. By substituting \eqref{eq:conditional_probability}, \eqref{entropy_QAM}, and \eqref{conditional_entropy_QAM} into \eqref{eq:SumRate,2}, we obtain the uplink multi-user sum-rate of QAM-modulated coherent states as
\begin{align}\label{eq:SumRate_appendix_QAM}
I_{\mathrm{sum}}^{\text{Q}}\!=&K\log m \!\!+\! \frac{1}{m^K\ln 2}\sum_{y=0}^{+\infty} \bigg\{\sum_{i=1}^{m^K} \Big[ \big| \langle \alpha_y | G^{-1/2} | \alpha_i \rangle  \big|^2\times \nonumber\\
& \ln \big| \langle \alpha_y | G^{-1/2} | \alpha_i \rangle  \big|^2\Big] - \sum_{i=1}^{m^K} \big| \langle \alpha_y | G^{-1/2} | \alpha_i \rangle  \big|^2\times\nonumber\\
& \ln \Big[ \sum_{i=1}^{m^K} \big| \langle \alpha_y | G^{-1/2} | \alpha_i \rangle  \big|^2 \Big]
\bigg\}.
\end{align}
Since \(G\) is a positive definite matrix with its minimum singular value satisfying $\sigma_{\min}^2>0$, applying the Cauchy-Schwarz inequality yields  \(\left| \langle \alpha_y | G^{-1/2} | \alpha_i \rangle \right|^2 \leq \sigma_{\min}^{-2} \|\alpha_y\|^2 \|\alpha_i\|^2 \leq C e^{-\lambda y}\), where $C>0$ and $\lambda>0$. As \(y\to\infty\),
\(\left| \langle \alpha_y | G^{-1/2} | \alpha_i \rangle \right|^2 \) decays exponentially. Hence, both \(\ln \big[ \left| \langle \alpha_y | G^{-1/2} | \alpha_i \rangle \right|^2 \big]\) and \(\ln \big[ \sum_{i=1}^{m^K} \left| \langle \alpha_y | G^{-1/2} | \alpha_i \rangle \right|^2 \big]\) converge; and so does~\eqref{eq:SumRate_appendix_QAM}.

\bibliographystyle{IEEEtran}
\bibliography{myreferences}
\end{document}